\documentclass[11pt]{article}
\pdfoutput=1
\usepackage[margin=1in]{geometry}

\DeclareMathAlphabet{\mathbbold}{U}{bbold}{m}{n}
\usepackage{dsfont}
\usepackage{amsmath,amsfonts,amssymb,amsthm}
\usepackage{mathtools}
\usepackage[usenames,dvipsnames,svgnames,table]{xcolor}
\usepackage{thm-restate}
\usepackage{enumitem}
\usepackage{microtype}
\usepackage{braket} 

\usepackage[nobiblatex]{xurl} 
\usepackage[pagebackref]{hyperref}
\hypersetup{
    pdftitle={No quantum speedup over gradient descent for non-smooth convex optimization}, 
    pdfauthor={Ankit Garg, Robin Kothari, Praneeth Netrapalli, Suhail Sherif}, 
    colorlinks=true, 
    linkcolor=blue, 
    citecolor=blue, 
    urlcolor=blue 
}

\renewcommand{\backref}[1]{}

\renewcommand{\backrefalt}[4]{%
\ifcase #1 %
\or
[p.\ #2]%
\else
[pp.\ #2]%
\fi}

\makeatletter
\newcommand{\para}{%
  \@startsection{paragraph}{4}%
  {\z@}{2ex \@plus 3.3ex \@minus .2ex}{-1em}%
  {\normalfont\normalsize\bfseries}%
}
\makeatother

\usepackage[nameinlink]{cleveref}

\newtheorem{theorem}{Theorem}
\newtheorem{lemma}[theorem]{Lemma}
\newtheorem{proposition}[theorem]{Proposition}

\theoremstyle{definition}
\newtheorem{definition}[theorem]{Definition}

\newtheorem{problem}[theorem]{Problem}

\newcommand{\R}{\mathbb{R}}
\newcommand{\E}{\mathop{{}\mathbb{E}}}

\newcommand{\dotp}[2]{\langle #1, #2 \rangle}
\newcommand{\abs}[1]{| #1 |}
\newcommand{\K}{\mathcal{K}}

\newcommand{\FO}{\mathcal{FO}}
\newcommand{\iprod}[2]{\left\langle #1, #2 \right \rangle}
\newcommand{\defeq}{\coloneqq}

\newcommand{\mV}{\mathcal{V}}
\newcommand{\mA}{\mathcal{A}}
\newcommand{\mF}{\mathcal{F}}
\newcommand{\F}{\mathcal{F}}

\newcommand{\wall}{\mathcal{W}}
\newcommand{\nem}{\mathcal{N}}
\newcommand{\eps}{\epsilon}

\newcommand{\mD}{\mathcal{D}}
\newcommand{\pmone}{\{-1,+1\}}

\newcommand{\spn}{\mathsf{span}}
\newcommand{\B}{\{0,1\}}
\newcommand{\PK}{\mathcal{P}_{\K}}

\newcommand{\tO}{\widetilde{O}}
\DeclareMathOperator{\tOmega}{\widetilde{\Omega}}

\newcommand{\floor}[1]{\left\lfloor{#1}\right\rfloor}

\newcommand{\norm}[1]{\|{#1}\|}
\newcommand{\Norm}[1]{\left\|{#1}\right\|}

\renewcommand{\(}{\left(}
\renewcommand{\)}{\right)}
\renewcommand{\>}{\rangle}

\DeclareRobustCommand{\noop}[1]{}

\DeclareMathOperator*{\argmin}{argmin}

\begin{document}

\title{No quantum speedup over gradient descent\\ for non-smooth convex optimization}

\author{
Ankit Garg\footnote{Microsoft Research India. Email: garga@microsoft.com.} 
\and
Robin Kothari\footnote{Microsoft Quantum and Microsoft Research, Redmond, WA, USA. \texttt{robin.kothari@microsoft.com}}
\and 
Praneeth Netrapalli\footnote{Microsoft Research India. Email: praneeth@microsoft.com}
\and
Suhail Sherif\footnote{Microsoft Research India and Tata Institute of Fundamental Research, Mumbai. Email: suhail.sherif@gmail.com.}
}
\date{}

\maketitle

\begin{abstract}
We study the first-order convex optimization problem, where we have black-box access to a (not necessarily smooth) function $f:\R^n \to \R$ and its (sub)gradient. Our goal is to find an $\eps$-approximate minimum of $f$ starting from a point that is distance at most $R$ from the true minimum. If $f$ is $G$-Lipschitz, then the classic gradient descent algorithm solves this problem with $O((GR/\eps)^{2})$ queries. Importantly, the number of queries is independent of the dimension $n$ and gradient descent is optimal in this regard: No deterministic or randomized algorithm can achieve better complexity that is still independent of the dimension $n$.

In this paper we reprove the randomized lower bound of $\Omega((GR/\eps)^{2})$ using a simpler argument than previous lower bounds. We then show that although the function family used in the lower bound is hard for randomized algorithms, it can be solved using $O(GR/\epsilon)$ quantum queries. We then show an improved lower bound against quantum algorithms using a different set of instances and establish our main result that in general even quantum algorithms need $\Omega((GR/\eps)^{2})$ queries to solve the problem. Hence there is no quantum speedup over gradient descent for black-box first-order convex optimization without further assumptions on the function family.
\end{abstract}

\section{Introduction}
\label{sec:intro}

The classic gradient descent algorithm, first proposed by Cauchy in 1847, is a popular algorithm for minimizing functions in high-dimensional spaces. For some problems, such as the case of convex function minimization that we consider in this paper, gradient descent provably converges to the function's global minimum. For other problems, such as finding good parameters of a deep neural network, gradient descent does not necessarily converge to a global minimum, and yet it has remarkable performance in practice.

Given the algorithm's popularity, it is interesting to ask if gradient descent can be sped up on a quantum computer. However, it's not obvious how to formalize this question since it's not clear what it means for a quantum algorithm to speed up a given classical algorithm. For example, the best known classical algorithm for integer factorization is the general number field sieve (GNFS). Does Shor's quantum algorithm for integer factorization speed up GNFS, or is it simply a different algorithm that solves the same problem?

One way to formalize the question \emph{Can quantum computers speed up gradient descent?} is to consider a computational problem that is provably solved by gradient descent, and for which gradient descent is optimal among all classical algorithms. We can then ask if quantum algorithms can solve this problem faster than gradient descent. The second condition, that gradient descent is optimal among classical algorithms, is required since otherwise quantum computers would trivially be able to outperform gradient descent by using the best classical algorithm.

Fortunately, there is a canonical optimization task that is solved optimally by gradient descent: convex optimization with black-box first-order oracles. A more thorough introduction to the theory of black-box convex optimization can be found in the textbooks by Nemirovsky and Yudin~\cite{nemirovsky1983problem}, Nesterov~\cite{Nes04,nesterov2018lectures}, and the monograph by Bubeck~\cite{Bub15}.

\subsection{First-order convex optimization}

Let's start with the unconstrained convex minimization problem for a convex function $f:\R^n \to \R$. Here we want to find an $x \in \R^n$ that's $\eps$-close to minimizing the function $f$. More precisely, if we let $x^* \defeq \argmin_{x\in\R^n} f(x)$, then our goal is to find any $x\in\R^n$ such that $f(x)-f(x^*)\leq \eps$.

To obtain algorithms that are very general, this problem is often studied in the setting of black-box optimization. Here we do not assume any particular structure of the function $f$ (e.g., that $f$ is a low-degree polynomial), and only assume that we have some efficient method of computing $f$ by an algorithm or circuit. In other words, we view $f$ as a black box.

If we only had access to a black-box computing $f$, this would be zeroth-order optimization. In first-order optimization, we additionally assume we can also compute the gradient of $f$, or more precisely, since the gradient may not exist, we assume we can compute some subgradient of $f$ (defined in \Cref{sec:prelim}). 
We call this oracle the \emph{first-order oracle} and denote it by $\FO(f)$. In this work we consider arbitrary convex functions that are not necessarily smooth,\footnote{In the optimization literature, a \emph{smooth function} is a function that is differentiable everywhere in its domain, so the gradient is well defined, and whose gradient has bounded Lipschitz constant.} and so we cannot assume that the gradient exists. Our goal is to solve the function minimization problem while minimizing the number of calls or queries to the black boxes for $f$ and some subgradient of $f$.

One might wonder why we consider queries to $f$ and the subgradient of $f$ to cost the same. This assumption is justified in many practical situations because of the \emph{cheap gradient principle}~\cite{GW08}, which says that ``the cost to evaluate the gradient $\nabla f$ is bounded above by a small constant times the cost to evaluate the function itself.'' 
This provably holds in many models of computation; E.g., for arithmetic circuits over $+$ and $\times$, it can be proved that the complexity of computing the gradient is at most 5 times the complexity of computing $f$~\cite{BS83}. 
The conversion of source code computing $f$ to code computing $\nabla f$ can often be done automatically in many programming languages, and such methods are called \emph{automatic differentiation} or \emph{algorithmic differentiation}~\cite{GW08}. 
These same principles essentially carry over to the computation of subgradients~\cite{KL18}. In the quantum setting, there is additional motivation to assume that a function and its gradient cost roughly the same since we can obtain the gradient (or a subgradient) of a function from a black-box computing the function, as shown in a sequence of papers that make increasingly weaker assumptions on the function oracle~\cite{Jor05,GAW19,CCLW20,vAGGdW20}.

Now that we have black-box access to $f$ and $\FO(f)$, we also need a starting point $x_0 \in \R^n$ to begin our search for a minimum. We require this to be an input, and the complexity will depend on  how close this is to $x^*$, since otherwise the interesting portion of the function where the minimum is achieved might be hiding in some small corner of $\R^n$ that we cannot efficiently locate with only black-box access. Since we can easily shift the function by a fixed vector, without loss of generality we assume $x_0 = \vec 0$ is the origin. 
Let the distance between $x_0 = \vec 0$ and $x^*$\footnote{If $x^*$ is not unique, we can let $R$ be the distance between $x_0$ and the closest $x^*$ to it.} be $R \defeq \norm{x^*}$.\footnote{%
Throughout this paper $\norm{\cdot}$ always denotes the standard $\ell_2$ norm in $\R^n$ defined as $\norm{z} \defeq \sqrt{\sum_i z_i^2}$.
}
For convenience, we will assume that $R$ is part of the input as well, although this can be relaxed by binary searching for the correct value of $R$.

Finally, it is also reasonable that the complexity of our algorithms depend on how quickly $f$ can change, since the value of $f$ at some point only constrains its values at nearby points if the function does not change too rapidly. Let $G$ be an upper bound on the Lipschitz constant of $f$ (defined in \Cref{sec:prelim}), and we assume this is part of the input as well.

We are now ready to formally define the first-order convex minimization problem in the black-box setting. We use $B(x,R)\defeq\{y:\norm{x-y}\leq R\}$ to denote an $\ell_2$-ball of radius $R$ around $x$.

\begin{problem}[First-order convex minimization]\label{prob:FOCM}
Let $f:\R^n \to \R$ have Lipschitz constant at most $G$ on $B(\vec{0},R)$, and let
\begin{equation}
  x^* \defeq \argmin_{x \in B(\vec{0},R)} f(x).
\end{equation}
Then given $n$, $G$, $R$, and $\eps>0$, the goal is to output a solution $x\in B(\vec{0},R)$ such that $f(x)-f(x^*)\leq \eps$ while minimizing the number of queries to $f$ and $\FO(f)$.\footnote{For simplicity, we assume that these oracles output real numbers to arbitrarily many bits of precision. Since the main results of this paper are lower bounds, this only makes our results stronger.} 
\end{problem}

Note that we allow algorithms to query the function and gradient oracles at any point in $\R^n$ even though the domain we are minimizing over is $B(\vec 0,R)$. 
This only makes our lower bounds stronger, and the algorithms discussed in this paper never query the oracles outside the domain.

As we discuss in \Cref{sec:prelim}, although the problem seems to involve 4 parameters, the parameters $G$, $R$, and $\eps$ are not independent since we can rescale the input and output spaces of $f$ and assume $G=1$ and $R=1$ without loss of generality. Thus any upper or lower bound on the complexity of this problem will be a function of $n$ and $GR/\eps$.

\subsection{Classical algorithms for first-order convex minimization}

Gradient descent, or in this case subgradient descent, is a simple algorithm that starts from a point $x_0$ and takes a small step (governed by a step size $\eta$) in the opposite direction of the subgradient returned at $x_0$. Intuitively this brings us closer to the minimum since we are stepping in the direction where $f$ decreases the most.

We can now describe the performance of subgradient descent for \Cref{prob:FOCM}. Since this is a constrained optimization problem, we use the projected subgradient descent algorithm, which is subgradient descent with the added step of projecting the current vector back onto the ball $B(\vec{0},R)$ after every step.

\begin{restatable}[Complexity of projected subgradient descent]{theorem}{PGD}
  \label{thm:PGD}
The projected subgradient descent algorithm solves \Cref{prob:FOCM} using $(GR/\eps)^2$ queries to $f$ and $\FO(f)$.
\end{restatable}

To be self contained, we prove this in \Cref{sec:prelim}. Observe that the query complexity of this algorithm, the number of queries made by the algorithm, is independent of $n$.\footnote{Of course, the time complexity of implementing this algorithm will be at least linear in $n$ since each query to either oracle requires us to manipulate a vector of length $n$.} 
This is quite surprising at first and partly explains why gradient descent and its variants are popular in high-dimensional applications. More generally, we call such algorithms \emph{dimension-independent algorithms}.

There also exist dimension-dependent algorithms for \Cref{prob:FOCM} that work well when $n$ is small. For example, the center of gravity method~\cite{Bub15} solves this problem with $O(n \log(GR/\eps))$ queries, which is very reasonable when $n$ is small (and the algorithm is very efficient in terms of $\eps$). In this work we focus on dimension-independent algorithms and assume that $n$ is polynomially larger than the other parameters in the problem. 

When $n$ is large, we cannot improve over projected subgradient descent (\Cref{thm:PGD}) using any deterministic or randomized algorithm.  We reprove the (well known) optimality of this algorithm among deterministic and randomized algorithms. This result is presented in \Cref{sec:randomizedlb}.  

\begin{restatable}[Randomized lower bound]{theorem}{RLB}
  \label{thm:randomizedlb}
  For any $G$, $R$, and $\eps$, there exists a family of convex functions $f:\R^n \to \R$ with $n = O((GR/\eps)^{2})$, 
  with Lipschitz constant at most $G$ on $B(\vec{0},R)$, such that any classical (deterministic or bounded-error randomized) algorithm that solves \Cref{prob:FOCM} on this function family must make $\Omega((GR/\eps)^2)$ queries to $f$ or $\FO(f)$ in the worst case.
\end{restatable}

This lower bound on query complexity has been shown in several prior works~\cite{nemirovsky1983problem,woodworth2017lower,BJLLS19}, but we believe our proof is simpler and the dimension $n$ required in our proof seems to be smaller than that in prior works. Note that while several expositions of gradient descent prove the lower bound for deterministic algorithms, very few sources establish a lower bound against randomized algorithms.

Our lower bound uses the following hard family of functions: For any $z \in \{-1,+1\}^n$, let $f_z(x_1,\ldots,x_n) = \max_{i \in [n]} z_i x_i$,\footnote{We use $[n]$ to denote the set of positive integers less than or equal to $n$, i.e., $[n] \defeq \{1,\ldots,n\}$.} 
 where $n=O(1/\eps^2)$.
These functions are convex with Lipschitz constant $1$. 
We show that finding an $\eps$-approximate minimum within $B(\vec 0 ,1)$ requires $\Omega(n)$ queries to the oracles. We establish the lower bound by showing that with high probability, every query of a randomized algorithm only reveals $O(1)$ bits of information about the string $z$, but an $\eps$-approximate solution to this problem allows us to reconstruct the string $z$, which has $n$ bits of information.

\subsection{Quantum algorithms for first-order convex minimization}

We then turn to quantum algorithms for solving \Cref{prob:FOCM}. At first, it might seem that since gradient descent is a sequential, adaptive algorithm where each step depends on the previous one, there is little hope of quantum algorithms outperforming gradient descent. 

On the other hand, consider the hard family of functions described above that witnesses the classical randomized lower bound in \Cref{thm:randomizedlb}. While this is hard for classical algorithms, we show in \Cref{sec:qalg} that there is a quantum algorithm that solves the problem on this family obtaining a quadratic speedup over any classical algorithm (and in particular, over gradient descent).

\begin{restatable}[Quantum algorithm for classically hard function family]{theorem}{QUB} 
  \label{thm:quantumub}
  There is a quantum algorithm that solves \Cref{prob:FOCM} on the class of functions that appear in the classical lower bound of \Cref{thm:randomizedlb} using $O(GR/\eps)$ queries to the oracle for $f$.
\end{restatable}

Notably, unlike most quadratic speedups in quantum computing, the source of this quadratic speedup is not Grover's algorithm or amplitude amplification. \Cref{thm:quantumub} uses Belovs' quantum algorithm for learning symmetric juntas, which is constructed by exhibiting a feasible solution to the dual semidefinite program of the negative-weights adversary bound~\cite{Bel14}.

Now that we have shown a quadratic quantum speedup on a family of instances known to be hard for classical algorithms, there is some hope that quantum algorithms may provide some speedup for the general first-order convex minimization problem. Alas, our next result (established in \Cref{sec:quantumlb}), which is our main result, shows that this is not the case, and quantum algorithms cannot in general yield a speedup over classical algorithms for first-order convex minimization.

\begin{restatable}[Quantum lower bound]{theorem}{QLB}
  \label{thm:quantumlb}
  For any $G$, $R$, and $\eps$, there exists a family of convex functions $f:\R^n \to \R$ with $n = \tO((GR/\eps)^4)$, with Lipschitz constant at most $G$ on $B(\vec{0},R)$, such that any quantum algorithm that solves \Cref{prob:FOCM} with high probability on this function family must make $\Omega((GR/\eps)^2)$ queries to $f$ or $\FO(f)$ in the worst case.
\end{restatable}

Our lower bound uses ideas from the lower bound against parallel randomized algorithms recently established by Bubeck, Jiang, Lee, Li, and Sidford~\cite{BJLLS19}.

At a high level, the hard family of functions used in the randomized lower bound does not work for quantum algorithms because although classical algorithms can only learn $O(1)$ bits of information per query, quantum algorithms can make queries in superposition and learn a little information about many bits simultaneously.
We remedy this by choosing a new family of functions in which with high probability, no matter what query the quantum algorithm makes, the oracle's response is essentially the same. 
This allows us to control what the quantum algorithm learns per query, but now the instance is more complicated and the quantum algorithm learns $O(n)$ bits of information per query. Since the final output of the algorithm is a vector in $\R^n$, we cannot use the argument used before that simply compared the information learned per query to the total information that needs to be learned.
Instead we use the venerable hybrid argument~\cite{BBBV97} to control what the quantum algorithm learns and show that it cannot find an $\eps$-approximate solution to the minimization problem.

\subsection{Related work}

Classically, there is a long history of the study of oracle complexity (also known as black-box complexity or query complexity) for deterministic and randomized algorithms for non-smooth and smooth convex optimization. The setting considered in this paper, first-order convex optimization, where the algorithm has query access to the function value and the gradient, is very well studied. This topic is too vast to survey here, but we refer the reader to \cite{nemirovsky1983problem,Nes04,nesterov2018lectures,Bub15} for more information about upper and lower bounds that can be shown in this setting. 

There also has been work in the classical parallel setting, where in each round the algorithm is allowed to query polynomially many points and the goal is to minimize the number of rounds \cite{Nem94, balkanski2018parallelization, diakonikolas2019lower, BJLLS19}. Our work is most closely related to this setting and borrows many ideas from these works. Although quantum algorithms and parallel classical algorithms are incomparable in power, the constructions used to thwart parallel classical algorithms in these papers also help with showing quantum lower bounds.

In the quantum setting, there has been some work on convex optimization in the oracle model. There is also work on quantum gradient descent not in the oracle model. For example, one situation studied is where the dimension $n$ of the optimization space is very large and the vectors are encoded in quantum states of dimension $\log n$. See \cite{RSW+19,kerenidis2020quantum} and the references therein for more information. 
Another setting is the work on semidefinite programming, an important special case of convex optimization, but these algorithms exploit the specific structure of semidefinite programs~\cite{BS17,vAGGdW17,BKL+19,vAG19} and are not directly related to our work.

While in the classical setting, in general, a function value oracle is weaker than a gradient oracle, this is not the case in the quantum setting. Given a function value oracle, one can get a gradient oracle quite efficiently (with an $\widetilde{O}(1)$ overhead) \cite{Jor05, GAW19, vAGGdW20, CCLW20}. A similar result also holds for simulating a separation oracle given a membership oracle for convex bodies \cite{vAGGdW20, CCLW20}. As discussed before, our focus in this paper is to see if quantum algorithms can outperform classical algorithms when given a function oracle and gradient oracle since in many relevant settings, gradient computation is cheap in the classical case as well. 

The most related works are the papers by Chakrabarti, Childs, Li, and Wu~\cite{CCLW20} and van Apeldoorn, Gily{\'{e}}n, Gribling, and de Wolf~\cite{vAGGdW20}. 
These papers establish very similar results so we cover them together. 
These papers study the problem of black-box convex optimization, and their results are phrased in the slightly different language of membership and separation oracles, but this is not the main difference between their work and our work. 
Indeed, it is possible to recast our problem in their setting (see the discussion in the introduction in \cite{vAGGdW20} for how to do this). 
The main difference is that their algorithms are dimension-dependent and have complexities that depend on $n$, whereas we're working in the parameter regime where $n$ is large and so we seek algorithms that are independent of $n$. 

Specifically, \cite{CCLW20} and \cite{vAGGdW20} consider the problem of minimizing a linear function over a convex body given via a membership or separation oracle. A membership oracle for a convex body tells us whether a given point $x$ is in the convex body and a separation oracle in addition when $x$ is not in the body outputs a hyperplane that separates $x$ from the convex body. 
Classically, the problem of outputting an $\eps$-approximate solution can be solved with $O(n^2 \, \mathrm{polylog}(\cdot))$ queries to a membership oracle, where we are suppressing polylogarithmic dependence on several parameters (including $\eps$). 
These two papers show a quantum algorithm that makes only $O(n \, \mathrm{polylog}(\cdot))$ membership queries. 
The key technical component of this is a construction of a separation oracle from a membership oracle with only polylogarithmic overhead. To do this, they first show how to obtain an approximate subgradient oracle from a function oracle with only polylogarithmic overhead.

There are also several lower bounds shown in these papers. In \cite{vAGGdW20}, the authors prove that quantum algorithms do not give any advantage over classical algorithms in the setting where we are not given a point inside the convex body to start with. This setting is not directly comparable to our setting, as far as we are aware. In the setting where we do know a point inside the convex body, which is very similar to our setting, \cite{vAGGdW20, CCLW20} prove a lower bound of $\Omega(\sqrt{n})$, which is quadratically worse than their algorithm. While, in general, their results are incomparable to our results, one specific comparison to our results is that \cite[Theorem 3.3]{CCLW20} essentially shows a $\tOmega(\min\{GR/\eps, \sqrt{n}\})$ lower bound on the number of oracle calls to a function value oracle for the setting in \Cref{prob:FOCM}.\footnote{This is equivalent to our setting, where we have a function value and gradient oracle, due to their results.} Note that this is quadratically worse than our tight lower bound (\Cref{thm:quantumlb}) in the dimension-independent setting (i.e., when the dimension $n$ is large compared to $GR/\eps$).

\subsection{Paper organization and summary of contributions}
We first present some preliminaries on convex optimization in \Cref{sec:prelim}. In \Cref{sec:randomizedlb} we reprove the lower bound for randomized algorithms (\Cref{thm:randomizedlb}) using a simpler argument compared to prior works. In \Cref{sec:qalg}, we show that quantum algorithms can solve the hard instance from \Cref{thm:randomizedlb} faster than randomized algorithms, obtaining a quadratic speedup (\Cref{thm:quantumub}). In \Cref{sec:quantumlb}, we present a different hard instance and show our main result that quantum algorithms cannot obtain any speedup over gradient descent for the first-order convex optimization problem (\Cref{thm:quantumlb}). We conclude with open problems in \Cref{sec:open}.

\section{Convex optimization preliminaries}
\label{sec:prelim}

As described in the introduction, we are interested in approximately minimizing a convex function $f: \R^n \rightarrow \R$ on some closed convex set $\K \subseteq \R^n$. 
A function $f:\R^n \to \R$ is convex if for all $x,y\in \R^n$ and $t \in [0,1]$, 
\begin{equation}
	t f(x) + (1-t) f(y) \geq f(tx + (1-t)y).
\end{equation}

A set $\K \subseteq \R^n$ is convex if the line segment joining two points in $\K$ is also contained in $\K$. We will consider convex sets of bounded size, and specifically let $2R$ be the diameter of $\K$, i.e., 	
\begin{equation}\label{eqn:diameter}
	\max_{x,y \in \K} \norm{x-y}\leq 2R,
\end{equation}
where $\norm{z} \defeq \sqrt{\sum z_i^2}$ is the Euclidean norm. 

It turns out that the query complexity of first-order convex optimization depends only on $R$ no matter how complicated the set $\K$ happens to be.
However, to obtain an algorithm with efficient time complexity we require that the set $\K$ be simple enough that we can efficiently implement a projection operator for $K$. This means given any $y \in \R^n$, we can efficiently compute $\PK(y) \in \K$, which satisfies $\norm{\PK(y)-y} = \min_{z \in \K} \norm{z-y}$.
Since the main result of this paper is a lower bound, our lower bound is stronger if shown for a simple convex set $\K$. So throughout this paper we work with the set $\K = B(\vec 0, R)$, the $\ell_2$-ball of radius $R$ around the origin. 

In the model of first-order black-box optimization, we have access to a black-box that computes the function $f$ on any input $x \in \R^n$. In addition to this, we also have a \emph{first-order oracle}, $\FO(f)$, which when queried at any point $x \in \R^n$ returns some vector $g_x \in \R^n$ that satisfies for all $y \in \R^n$, 
\begin{align}\label{eqn:subgrad}
	f(y) \geq f(x) + \iprod{g_x}{y-x}.
\end{align}
Since $f$ is convex, it is known that such a vector $g_x$ exists for all $x \in \R^n$~\cite{Nes04}. Any vector $g_x$ satisfying~\eqref{eqn:subgrad} is called a subgradient of $f$ at $x$, and the set of all subgradients at $x$ is called the subdifferential at $x$ and denoted by $\partial f(x)$. 
If $f$ is differentiable at $x$ then $g_x$ is unique and equal to $\nabla f(x)$, the gradient of $f$ at $x$, defined as 
\begin{equation}
	\nabla f(x) \defeq \left(\frac{\partial f(x)}{\partial x_1}, \ldots, \frac{\partial f(x)}{\partial x_n}\right).
\end{equation}
We will say that the function $f$ has Lipschitz constant at most $G$ in $\K$ if $\norm{g_x} \leq G$ for every $x \in \K$.\footnote{%
This is slightly different from the usual definition of the Lipschitz constant where we would say $f$ is $G$-Lipschitz in $\K$ if for all $x,y \in \K$, $\abs{f(x) - f(y)} \leq G \norm{x-y}$. 
Our definition is the same as requiring the function $f$ to be $G$-Lipschitz according to this definition in an open set that contains $\K$.}

As described in \Cref{prob:FOCM}, we are interested in algorithms that take as inputs the parameters $n$, $G$, $R$, and $\epsilon > 0$, and have access to $f$ and a first-order oracle $\FO(f)$, and output ${x} \in B(\vec 0, R)$ such that $f({x}) - f(x^*) \leq \epsilon$, where $x^* \defeq \argmin_{x \in B(\vec 0, R)} f(x)$. 

In the quantum setting, we have quantum analogues of these oracles. There is a straightforward generalization of any oracle to the quantum setting, which makes the classical oracle reversible and then allows queries in superposition to this oracle. This quantum generalization of the oracle is justified by the fact that if we had a classical circuit or algorithm computing a function $f$, then it is possible in a completely black-box manner to construct the quantum oracle corresponding to $f$, and this oracle would then support superposition queries. We discuss quantum oracles in more detail in \Cref{sec:quantumlb}, but for now it is sufficient to consider them as computing the same functions as the classical oracles, except that they can additionally be queried in superposition.

Note that it is sufficient to consider the special case of the problem where $G=R=1$. While this seems like a special case, given an $f$ and $\K$ with Lipschitz constant $G$, radius $R$, and optimization accuracy $\epsilon$, we can instead minimize $\hat{f}(x) \defeq \frac{1}{GR} f\left(R{x}\right)$ over $\hat{\K} \defeq \K / R$, which have Lipschitz constant and radius $1$ up to an accuracy of $\frac{\epsilon}{GR}$. So we consider $G=R=1$ without loss of generality, or for general $G$ and $R$, the complexity must be a function of $GR/\eps$.

The query complexity of an algorithm that solves \Cref{prob:FOCM} is the maximum number of oracle calls it makes for fixed values of $n$, $G$, $R$, and $\eps$, where the maximum is taken over all convex functions $f$ with Lipschitz constant at most $G$, and all first order oracles $\FO(f)$ for $f$ (i.e., the algorithm must work for any choice of first-order oracle that correctly outputs some subgradient of $f$ at $x$). As discussed, the query complexity must be a function of $n$ and $GR/\eps$. Furthermore, since we're interested in dimension-independent algorithms, we study algorithms that only depend on $GR/\eps$ and not on $n$.

Given a class of algorithms, such as deterministic, randomized, or quantum algorithms, the query complexity of first-order Lipschitz convex optimization for that class of algorithms is the minimum query complexity of any algorithm in that class that solves \Cref{prob:FOCM}.

As we show in \Cref{thm:randomizedlb} in \Cref{sec:randomizedlb}, the randomized query complexity of this problem (and hence the deterministic query complexity) is at least $\Omega((GR/\eps)^2)$ in the dimension-independent setting.

In the rest of this section, we will prove that the deterministic query complexity of \Cref{prob:FOCM} is $O((GR/\eps)^2)$, matching the (randomized) lower bound of \Cref{thm:randomizedlb}. In particular, we describe how the well-known gradient descent algorithm, or more precisely a variant known as the projected subgradient descent algorithm, achieves this upper bound. We now restate \Cref{thm:PGD} for convenience:

\PGD*

\begin{proof}
Without loss of generality we assume $G=R=1$. The projected subgradient descent algorithm is easy to describe. We start by setting the initial vector $x_0 = \vec 0$. The algorithm then computes $x_{t+1}$ from $x_t$ using the formula
\begin{align}\label{eqn:GD}
	x_{t+1} = \PK(x_t - \eta \cdot g_{x_t}),
\end{align}
where $\eta > 0$ is the step size, a parameter of the algorithm that we must choose, and $\PK$ is the projector onto $B(\vec 0,1)$.
After $T$ steps, the algorithm outputs $\hat{x}_T \defeq \frac{1}{T} \sum_{t=0}^{T-1} x_t$. To obtain the claimed upper bound we set the step size $\eta = \eps$.

Now we claim that for any $T\geq 1/\eps^2$, the output $\hat{x}_T$ satisfies:
\begin{align}
	f(\hat{x}_T) - f(x^*) \leq \eps.
\end{align} 
We prove this using the potential function $\norm{x_t - x^*}^2$. We have
\begin{align}
		\norm{x_{t+1} - x^*}^2 &= \norm{\PK(x_t - \eta g_{x_t}) - x^*}^2 \leq \norm{x_t - \eta g_{x_t} - x^*}^2 \\
		&= \norm{x_t - x^*}^2 - 2 \eta \iprod{
			g_{x_t}}{x_t - x^*} + \eta^2 \norm{g_{x_t}}^2, 
\end{align}
where the inequality uses the fact that projecting a vector outside $\K$ to $\K$ can only reduce its distance to a point in $\K$. We then use the Lipschitz condition ($\norm{g_{x_t}}^2\leq 1$) and the definition of the subgradient in \cref{eqn:subgrad} to get
\begin{align}
	\norm{x_{t+1} - x^*}^2 &\leq \norm{x_t - x^*}^2 - 2 \eta \left(f(x_t) - f(x^*)\right) + \eta^2.
\end{align}
	Taking a telescopic sum and averaging, we obtain
	\begin{align}
		\left(\frac{1}{T} \sum_{t=0}^{T-1} f(x_t) \right)- f(x^*) 
		\leq \frac{\norm{x_0 - x^*}^2 - \norm{x_{T-1} - x^*}^2}{2 \eta T} + \frac{\eta}{2} \leq \frac{1}{2 \eta T} + \frac{\eta}{2} \leq \eps,
	\end{align}
where the second inequality used the fact that $\norm{x_0 - x^*} = \norm{x^*} \leq R = 1$.
By convexity of $f$, $f(\hat{x}_T) \leq \frac{1}{T} \sum_{n=0}^{T-1} f(x_t)$, which proves the result.
\end{proof}

Note that although we stated and proved this for $\K = B(\vec 0, R)$, the upper bound on the number of queries made to the oracles holds for any $\K$ that is contained in $B(\vec 0 , R)$. However, if we wanted to implement this algorithm, then the time complexity would depend on how hard it is to implement the operator $\PK$, which projects onto the set $\K$.


\section{Randomized Lower Bound}
\label{sec:randomizedlb}

In this section, we prove a lower bound for randomized first-order methods for non-smooth convex optimization, restated here for convenience:

\RLB*

This lower bound is known and multiple proofs can be found in the literature \cite{nemirovsky1983problem, woodworth2017lower}. Our proof is elementary and we did not find it written anywhere, although it is conceptually similar to the one in \cite{nemirovsky1983problem}, and so we include it here for completeness. Our proof also has the dimension $n=\Theta(1/\eps^2)$, without any log factors, which is the best possible. As far as we are aware, the previous proofs required larger dimension. As we will see later, the family of instances used  is also interesting because we can get a quantum speedup for it, because of which we have to look at other instances to prove the quantum lower bound.

We can now define the family of convex functions used in the lower bound. For any $\epsilon>0$, we set $n = \floor{.9/\epsilon^2}$ and look at the following class of functions.

\begin{definition}
    Let $z \in \pmone^n$. Let $f_z: \R^n \rightarrow \R$ be defined as 
    \begin{equation}
        f_z(x_1,\dots,x_n) = \max_{i \in [n]} z_ix_i.
    \end{equation}
\end{definition}

Each such function is convex since it is a maximum of convex functions \cite[Theorem 3.1.5]{Nes04}. Note that if $f_z(x) = z_ix_i$ for some $i \in [n]$, then $z_ie_i$ is a subgradient of $f_z$ at $x$ (since $f_z(x) + \dotp{z_ie_i}{y-x} = z_iy_i \leq f_z(y)$). Hence the function is $1$-Lipschitz. We can also see that within the unit ball the function is minimized at the point 
\begin{equation}
   x^* =  \frac{-1}{\sqrt{n}} \sum_{i\in[n]} z_i e_i,
\end{equation}
and $f_z(x^*)=-1/\sqrt{n}$. Clearly given $x^*$ we can recover $z$ from it. We now show $z$ can even be recovered from an $\epsilon$-approximate minimum of $f_z$.

\begin{lemma}\label{lem:recoverz}
    Let $x$ be such that $f_z(x)-f_z(x^*)\leq \eps$. Then we can recover $z \in \B^n$ from $x \in \R^n$.
\end{lemma}

\begin{proof}
    Let $s_x \in \{-1,+1\}^n$ be the vector with $(s_x)_i = \mathrm{sign}(x_i)$, where $\mathrm{sign}(a) = +1$ if $a\geq 0$ and $\mathrm{sign}(a) = -1$ otherwise. We claim that $z = -s_x$. Toward a contradiction, if $(s_x)_i \neq -z_i$ for some $i$, then $(s_x)_i = z_i$, since these only take values in $\{-1,+1\}$. In this case, $x_i$ and $z_i$ agree in sign, and hence $f_z(x) \geq z_ix_i \geq 0$. Since $\epsilon < 1/\sqrt{n}$ (because of our choice of $n$ above) the point $x$ cannot satisfy $f_z(x)-f_z(x^*)\leq \eps$. 
\end{proof}

Since this function is not differentiable everywhere, for our lower bound we need to specify the behavior of the subgradient oracle on all inputs. The function is not differentiable only at $x \in \R^n$ where the maximum is achieved at multiple indices. 
In this case, the subgradient oracle responds as if the maximum was achieved on the smallest such index $i$, i.e., it responds with $z_ie_i$.
Note that for this function, querying the subgradeint oracle allows us to simulate a call to the function oracle as well, since the response is $z_i e_i$ for the index $i$ that achieves the maximum, so the function evaluates to $z_ix_i$ at that point, which we can compute since we know $x$.
So we can assume without loss of generality that an algorithm only queries the subgradient oracle.

Now that the problem is fully specified, we will show that any randomized optimization algorithm using the function oracle and this subgradient oracle will require $\Omega(n)$ queries in order to solve \Cref{prob:FOCM} with a constant probability of success.

The following will be the crux of the lower bound. Let $I \subseteq [n]$. We say a distribution $\mD$ over $\pmone^n$ is $I$-fixed if for $z \sim \mD$ the random variable $z_I$ is fixed and $z_{\overline{I}}$ is uniform over $\pmone^{\overline{I}}$.

\begin{lemma}\label{lem:fixeddistributions}
    Let $z$ be distributed according to an $I$-fixed distribution. Let $x$ be an arbitrary query made to the $f_z$ oracle. After one query to the subgradient oracles, the conditional distribution on $z$ given the answer is $I'$-fixed with $I \subseteq I'$ and $\E[|I'|] \leq |I| + 2$.
\end{lemma}

\begin{proof}
    Let $x$ be the algorithm's query. 
    The index $i$ that achieves the maximum in the definition of $f_z(x)$ can be computed as follows. 
    Let $i_1, \dots, i_n$ be the ordering of the indices $1$ to $n$ in decreasing order of $|x_i|$, with ties broken with the natural ordering on integers. 
    The oracle outputs $f_z(x) = z_{i_j}x_{i_j}$ and chooses the subgradient $z_{i_j}e_{i_j}$ where $j$ is the smallest index for which $x_{i_j}$ agrees in sign with $z_{i_j}$, and if no such index exists, then $j=n$. 
    
    Since $f_z(x)$ can be computed given the subgradient $z_{i_j}e_{i_j}$, the only information obtained from a query is the prefix $\{z_{i_k}\}_{k \leq j}$. 
    In other words, if the subgradient oracle responds with $z_{i_j}e_{i_j}$, then we have learned that for all indices $k\leq j$, we must have $\mathrm{sign}(x_i)=-z_i$, but we have not learned any more since the oracle's output does not depend on the bits of $z$ with index $i_k$ with $k>j$.
    After this query, we know the bits $z_{i_k}$ with $k \leq j$, but conditioned on these, the distribution on the remaining bits of $z$ continues to be uniform. 
    This is an $I'$-fixed distribution with $I' = I \cup \{i_k\}_{k \leq j}$.
    Intuitively, $I'$ cannot be much larger than $I$ since an index $i_k$ is part of this set only if the algorithm correctly guessed the sign of $z_{i_k}$ for this index and all indices with a smaller value of $k$.
    Since the initial distribution $z$ was uniformly at random outside of $I$ and $x$ is fixed, the probability of correctly guessing the first index (according to the $i_j$ ordering) that was not fixed is $1/2$, the probability of guessing the first two is $1/4$ and so on. 
    Thus the expected number of new entries fixed by one query is $\sum_{k=1}^{n \setminus |I|} k \cdot \frac{1}{2^k} \leq 2$.
\end{proof}

We can use this to show establish the final claim.

\begin{lemma}
    Let $z$ be sampled uniformly at random from $\pmone^n$. If a randomized algorithm $\mA$ outputs an $x$ with $f_z(x)-f_z(x^*)\leq \epsilon$ with probability at least $2/3$, then its query complexity is at least $n/3 - 1$.
\end{lemma}

\begin{proof}
    When $\mA$ outputs a point $x$, we will require it to also query the oracle at $x$ to see if it is indeed $\epsilon$-optimal. This can increase its query complexity by at most one. Let the query complexity of this modified $\mA$ be $t$. Whenever $\mA$ does output an $\epsilon$-optimal point, \Cref{lem:recoverz} implies that the conditional distribution on $z$ is $[n]$-fixed. For each $i \in [0,..,t]$, let $I_i$ be the random variable such that the distribution on $z$ after $i$ queries of $\mA$ is $I_i$-fixed (\Cref{lem:fixeddistributions} implies that after any sequence of queries it will be an $I$-fixed distribution for some $I$). Since $z$ is sampled uniformly at random from $\pmone^n$, $I_0 = \emptyset$. And since we want the algorithm to succeed with probability at least $2/3$, $\E[\abs{I_t}] \geq 2n/3$.

    However, $\abs{I_t} = \sum_{i=1}^{t} \abs{I_i} - \abs{I_{i-1}}$, and it is a simple consequence of \Cref{lem:fixeddistributions} that $\E[\abs{I_i} - \abs{I_{i-1}}] \leq 2$ for all $i$. So by the linearity of expectation, $\E[\abs{I_t}] \leq 2t$ and hence $t \geq n/3$.
\end{proof}

This proves a lower bound of $\Omega(1/\eps^2)$ on the randomized query complexity of first-order convex minimization for a function with $G=R=1$. As noted earlier, this is without loss of generality and implies the more general bound in \Cref{thm:randomizedlb}.

\subsection{Quantum speedup}
\label{sec:qalg}

In this section we prove \Cref{thm:quantumub}, restated for convenience:

\QUB*

The quantum speedup for the above class of functions relies on Belovs' quantum algorithm for Combinatorial Group Testing~\cite{Bel14}. Belovs showed that given access to an oracle making $\mathsf{OR}$ queries to an $n$-bit string, the $n$-bit string can be learned in $O(\sqrt{n})$ quantum queries. More formally, Belovs showed the following~\cite{Bel14}.

\begin{theorem}
    Let $x\in\B^n$ and $O_x$ be the unitary that for every $S\subseteq [n]$ and $b \in \B$, satisfies $O_x|S\>|b\>=|S\>|b \oplus \mathsf{OR}_x(S)\>$, where $\mathsf{OR}_x(S)=1$ if there is an $i \in S$ such that $x_i = 1$, and $\mathsf{OR}_x(S) = 0$ otherwise. 
    Then we can learn $x$ with high probability with  $O(\sqrt{n})$ quantum queries to the oracle $O_x$.
\end{theorem}

We can now prove \Cref{thm:quantumub}. 

\begin{proof}[Proof of \Cref{thm:quantumub}]
    In our optimization problem, making the query $x = \frac{1}{\sqrt{n}} \sum_{i \in S} e_i$ to the function oracle returns $f_z(x) = \frac{1}{\sqrt{n}}$ if there is an $i \in S$ such that $z_i = 1$. If there is no such $i \in S$, then it will output $f_z(x) = 0$, unless $S=n$, in which case it will output $-\frac{1}{\sqrt{n}}$.

    Hence a function value oracle for $f_z$ can be used to make $\mathsf{OR}$ queries to the string $z$, since it outputs $1$ if there is an $i \in S$ such that $z_i = 1$ and outputs $0$ (or $-1/\sqrt{n}$) otherwise. Using Belovs' algorithm, with $O(\sqrt{n})$ such queries, we can learn the locations of all the $1$s in $z$, which allows us to learn $z$ completely.
\end{proof}

This quantum algorithm is also essentially optimal for this problem and it is not hard to show an $\Omega(\sqrt{n}/\log n)$ lower bound for quantum algorithms. A similar lower bound is shown in \cite[Theorem 3.3]{CCLW20}, and we sketch a simpler proof of the claim here.

As discussed in the classical lower bound, what the subgradient oracle allows us to do is have a non-standard query to the unknown string $z \in \pmone^n$. In this non-standard query, we get to order the bits of $z$ however we like, and then submit a string in $\pmone^n$ and ask for the first index (according to our ordering) where our string agrees with $z$. As we showed in the classical lower bound, if we solve the optimization problem, then we also learn $z$. 

So we are left with answering the question of how hard it is to learn $z$ given these non-standard queries to $z$. Given standard queries to $z$, where we can only query one bit of our choice, it is well known that we need $\Omega(n)$ queries to learn $z$. But our non-standard query is easy to implement using Grover's algorithm with only $O(\sqrt{n})$ standard queries, since all we have to do is find the first bit of $z$ according to a known ordering where the queried string and $z$ agree. If the problem of learning $z$ with these non-standard queries used $T$ non-standard queries, then we could implement the non-standard queries ourselves with cost $O(\sqrt{n})$ and compose the two algorithms to obtain an algorithm for learning $z$ using standard queries with complexity $O(T\sqrt{n}\log n)$. (We have an additional log factor because we are composing two bounded-error algorithms.) Since this problem has a lower bound of $\Omega(n)$, we get $T=  \Omega(\sqrt{n}/\log n)$. It might be possible to remove this log factor using standard techniques for log factor removal (composing solutions of the dual of the adversary bound), but we have not attempted to work out the details. 
\section{Quantum lower bound}
\label{sec:quantumlb}

In this section, we show that for any $\eps$, there exists a $1$-Lipschitz family of functions such that any quantum algorithm that solves \Cref{prob:FOCM} on the unit ball must make $\frac{1}{100\epsilon^2}$ queries. In other words, there is no quantum first-order convex optimization algorithm that always outperforms the classical gradient descent algorithm described in \Cref{thm:PGD}. The function we will use was introduced by Nemirovsky and Yudin~\cite{nemirovsky1983problem}. To show the quantum lower bound, we adapt to the quantum setting the lower bound strategy of Bubeck et al.~\cite{BJLLS19} in the model of parallel algorithms.

We restate the main result proved in this section for convenience:

\QLB*

We start by first proving a qualitatively similar, but simpler result with a larger value of $n = \tO((GR/\eps)^6)$ in \Cref{subsec:quantumlowerbound}. If we only care about the optimality of gradient descent in the dimension-independent setting, this lower bound is sufficient. But if we also want to understand the trade-off between dimension-independent and dimension-dependent algorithms, then we would like to show this lower bound with as small a value of $n$ as we can. In \Cref{sec:wall}, we improve the lower bound to achieve the value of $n$ stated in this theorem. 

\subsection{Function family and basic properties}

We start by defining the family of functions $\F=\{f:\R^n \to \R\}$ that we use. The function family $\F$ depends on the dimension $n$ and two other parameters $k$ and $\gamma$. 
Since the function family we choose depends on $\eps$, the parameters $n$, $k$, and $\gamma$ will be functions of $\eps$. Our choice of $n$, $k$, and $\gamma$ will become clear later, but for now we simply choose them as follows. Let
\begin{align}\label{eq:choice}
    k \defeq \frac{1}{100\eps^2} \implies \eps = \frac{1}{10\sqrt{k}}\quad \mathrm{and} \quad \gamma \defeq \frac{1}{10k^{3/2}} = 100\eps^3. 
\end{align}
We choose $n$ such that it satisfies 
\begin{equation}\label{eq:gamma}
    \gamma \geq 8\sqrt{\frac{\log n}{n}} \implies n \defeq O\(\frac{\log(1/\eps)}{\eps^6} \) = \tO\(\frac{1}{\eps^6}\).
\end{equation}
The discussion before \Cref{lem:propf} explains the choice of $k$ and the discussion after \Cref{lem:propmin} explains the choice of $\gamma$. For the dimension $n$, see the discussion at the beginning of \Cref{sec:prob}.

We now define the function family for these specific choices of $n$, $k$, and $\gamma$.

\begin{definition}[Hard function family]\label{def:function}
    Let $\mV = \{(v_1,\ldots,v_k)  \mid \forall i,j, \in [k], \<v_i,v_j\> = \delta_{ij} \}$ be the set of all $k$-tuples of orthonormal vectors in $\R^n$. Let the family of functions $\mF = \{f_{V}\}_{V \in \mV}$  be defined as
    \begin{align}
        f_{(v_1,v_2,\dots,v_k)}(x) & \defeq \max_{i \in [k]} \bigl\{g_V^{(i)}(x)\bigr\}, \text{ where } g_V^{(i)}(x) \defeq \dotp{v_i}{x} + (k - i)\gamma \|x\|.
    \end{align}
\end{definition}

We will show that any quantum algorithm that solves \Cref{prob:FOCM} on the functions in this family must make $k$ queries. As we will prove, informally what happens is each query of the quantum algorithm to the gradient oracle only reveals a single direction $v_i$ to the algorithm. 
In fact, with very high probability the vectors are revealed in order, so that the algorithm first learns $v_1$, then $v_2$, and so on. 
As we will show in \Cref{lem:propmin}, any $\eps$-optimal solution must overlap significantly with all $v_i$, and thus any quantum algorithm must make $k$ queries. 
Since we want to show an $\Omega(1/\eps^2)$ bound, we choose $k$ to be a small multiple of $1/\eps^2$, which explains our choice for $k$ in \cref{eq:choice}. 

We now establish some basic properties of these functions.

\begin{lemma}[Properties of $f_V$]\label{lem:propf}
    For any $V \in \mV$, let $f_V$ and $g_V^{(i)}$ be as in \Cref{def:function}. 
    Then $f_V$ is convex with Lipschitz constant at most $1+k\gamma \leq 2$ on $B(\vec 0,1)$, and
    \begin{align}        
        &\text{for } x \neq \vec 0, \quad  \nabla \smash{g_V^{(i)}} (x) = v_i + (k-i)\gamma {x}/{\norm{x}}, \text{ and}\\
        &\text{for } x = \vec 0, \quad  \partial g_V^{(i)}(\vec 0) = \{v_i + (k-i)\gamma u \mid u \in B(\vec 0,1)\}, \text{ and}\\
        &\text{for any $x$}, \quad  \partial f_V(x) = \mathrm{ConvexHull}\bigl(\{u \in \partial g_V^{(i)}(x) \mid g_V^{(i)}(x)=f_V(x)\}\bigr),
    \end{align}
    where the convex hull of a set of vectors is the set of all convex combinations of vectors in the set.
    Lastly, for any $\alpha>0$, $f_V(\alpha x) = \alpha f(x)$ and $\partial f_V(\alpha x) = \partial f_V(x)$.
\end{lemma}
\begin{proof}
For all $V \in \mV$,  $f_V:\R^n \to \R$ is convex. This follows because linear functions and norms are convex functions~\cite[Example 3.1.1]{Nes04}, and the sum or maximum of convex functions is convex~\cite[Theorem 3.1.5]{Nes04}.

Let us now compute the subgradients of $g_V^{(i)}(x) =  \dotp{v_i}{x} + (k - i)\gamma \|x\|$. The linear function $\dotp{v_i}{x}$ is differentiable and its gradient is simply $v_i$. The Euclidian norm $\norm{x}$ is differentiable everywhere except at $x = \vec 0$. At $x\neq \vec 0$, the gradient of $\norm{x}$ is $x/\norm{x}$ and at $x=0$, the set of subgradients is $B(\vec 0,1)$~\cite[Example 3.1.5]{Nes04}. We also know that $\partial (\alpha_1 f_1(x) + \alpha_2 f_2(x)) = \alpha_1 \partial f_1(x) + \alpha_2 \partial f_2(x)$~\cite[Lemma 3.1.9]{Nes04}, which gives us the expressions for the subgradients of $g_V^{(i)}$.

For a function that is the maximum of functions $g_V^{(i)}$, we know that the set of subgradients is simply the convex hull of subgradients of those $g_V^{(i)}$ which achieve the maximum at the given point $x$~\cite[Lemma 3.1.10]{Nes04}.

The Lipschitz constant of a function is the maximum norm of any subgradient of the function. Since any vector in $\partial g_V^{(i)}$ has norm $1+k\gamma$, and any vector in $\partial f_V$ is the convex combination of vectors with norm at most $1+k\gamma$, the Lipschitz constant of $f_V$ is at most $1+k\gamma \leq 2$.

Finally, it is easy to see from the definition of $f_V$ that for $\alpha>0$, $f_V(\alpha x) = \alpha f(x)$ since each term in the max gets multiplied by $\alpha$. For $\partial f_V(\alpha x)$, note that this is a convex combination of $\partial g_V^{(i)}(\alpha x)$, and these do not depend on $\alpha$.
\end{proof}

For convenience we work with this family of functions with Lipschitz constant at most $2$ instead of $1$, which doesn't change the asymptotic bounds since we could just divide every function $f_V$ by $2$.

The last property essentially says that querying the function or its subgradient on a scalar multiple of a vector $x$ gives us only as much information as querying it on $x$. Thus we can assume that an algorithm only queries the oracles within the unit ball without loss of generality.

Now let us discuss the vector $x^* \in B(\vec 0,1)$ that minimizes $f_V(x)$ and vectors that $\eps$-approximately solve the minimization problem. First note that if $\gamma$ were equal to $0$, then the function would simply be $\max_{i \in [k]} \<v_i,x\>$, which requires us to minimize the component of $x$ in $k$ different directions subject to it being a unit vector. The solution to this is simply $\frac{-1}{\sqrt{k}} \sum_i  v_i$. Now $-1/\sqrt{k} = -10\eps$, so the overlap of $x$ with each direction $v_i$ is a large multiple of $\eps$. So even an $\eps$-approximate solution must have reasonable overlap with each of the vectors $v_i$. Specifically, each overlap must be at least $-9\eps$. Now in our function $f_V$ the term $\gamma$ is not $0$, but that term at most perturbs the function by $k\gamma = \eps$, which again is much smaller than $10\eps$, and thus even approximate solutions must have significant overlaps with all $v_i$. We formalize these properties below.

\begin{lemma}[Properties of the minimum]\label{lem:propmin}
    For any $V \in \mV$, let $f_V:\R^n \to \R$ be the function in \Cref{def:function} and let  $x^* \defeq \argmin_{x \in B(\vec 0,1)} f_V(x)$. Then $f_V(x^*) \leq -9\eps$. Furthermore, any $x\in \R^n$ that satisfies $|f_V(x)-f_V(x^*)|\leq \eps$ must satisfy for all $i \in [k]$, $\<v_i,x\> \leq -8\eps$.    
\end{lemma}
\begin{proof}
    Consider the vector $y= \frac{-1}{\sqrt{k}} \sum_{i \in [k]} v_i$. This is a vector in $B(\vec 0,1)$, satisfying $f_V(y) \leq \frac{-1}{\sqrt{k}} + (k-1)\gamma \leq \frac{-1}{\sqrt{k}} + k\gamma = -10\eps + \eps = -9\eps$, because we have $10\eps = \frac{1}{\sqrt{k}}$ and $k\gamma = \frac{1}{10\sqrt{k}} = \eps$. Thus $f_V(x^*) \leq f_V(y) \leq -9\eps$.

    Now consider any vector $x$ with $|f_V(x)-f_V(x^*)|\leq \eps$, which implies $f_V(x) \leq -8\eps$. If $\<v_i,x\> > -8\eps$ for any $i\in[k]$, then $f_V(x) \geq \<v_i,x\> + (k-i)\gamma \norm{x} > -8\eps$, which is a contradiction.
\end{proof}

This result crucially uses the relation between $\gamma$ and $k$ and because we want $k\gamma$ to be a constant factor (say $10$) smaller than $\sqrt{1/k}$, this informs our choice of $\gamma$ in \cref{eq:choice}. Our choice of $n$ in \cref{eq:gamma} will be discussed in the next section.

\subsection{Probabilistic facts about the function family}
\label{sec:prob}

So far all the properties we have discussed of our function family hold for any $V \in \mV$, but now we want to talk about a hard distribution over such functions. Specifically we want to talk about choosing a uniformly random (according to the Haar measure) $V$ from the infinite set $\mV$. 
It is easy to see how to sample a random $V$ once we can sample unit vectors from a subspace.
We start by choosing $v_1$ to be a Haar random unit vector from $\R^n$, let $v_2$ be a Haar random unit vector from $\spn(v_1)^{\perp}$, and so on, until $v_k$ is a Haar random unit vector in $\spn(v_1, v_2,\dots, v_{k-1})^{\perp}$. 
In the following, to improve readability, we will use boldface to denote random variables.

We can now discuss what determines our choice of $n$. By construction, the family of functions $\F$ has the property that if the input vector $x$ has equal inner product with all vectors $v_i$, then the maximum will be achieved uniquely on the first term $i=1$ because the additive term $(k-i)\gamma\norm{x}$ is largest for $i=1$. Now what we want to ensure is that this property holds even when $x$ does not have equal inner product with all $v_i$, but $x$ is chosen uniformly at random from $B(\vec 0,1)$. Or equivalently, we want this property to hold when $x$ is fixed, but the set $V$ is chosen uniformly at random. 

In either case, the inner product of $x$ with a random unit vector $v$ will be a random variable with mean $0$ due to symmetry. But the expected value of $|\<v,x\>|^2$ for a random unit vector $v$ is $1/n$, and in fact it will be tightly concentrated around $1/n$.  The following proposition follows from \cite[Lemma 2.2]{Ball97}.

\begin{proposition}\label{prop:concentration}
Let $x\in B(\vec 0, 1)$. Then for a random unit vector $v$, and all $c>0$,
    \begin{equation}
        \Pr_{v}(|\<x,v\>| \geq c) \leq 2e^{-nc^2/2}.
    \end{equation}
\end{proposition}

We choose $\gamma$ so that it is very unlikely (polynomially small in $n$) that the maximum is not achieved at $i=1$. From \Cref{prop:concentration}, we see that the probability of any $|\<v_i,x\>|^2$ being larger than a constant multiple of $\log n/n$ is inverse polynomially small. So it is sufficient to take $\gamma^2$ to be a large constant multiple of $\log n /n$ as in \cref{eq:gamma}.

In our lower bound we will need a slightly stronger result. We can show that if the vectors $v_1,\ldots,v_{t-1}$ are fixed (and hence known to the algorithm), and the remaining vectors $v_t, \ldots, v_k$ are chosen uniformly at random such that the set of vectors $\{v_1,\ldots,v_k\}$ is orthonormal, then the maximum will be achieved in the set $[t]$ with high probability. This generalizes the previous claim, which is the case of $t=1$, where none of the vectors were fixed.

\begin{lemma}[Most probable argmax]\label{lem:argmax}
    Let $1\leq t \leq k$ be integers and $\{v_1,\ldots,v_{t-1}\}$ be a set of orthonormal vectors. Let $\{v_t,\ldots,v_k\}$ be chosen uniformly at random so that the set $\{v_1,\ldots,v_k\}$ is orthonormal. Then 
\begin{equation}
    \forall x \in B(\vec 0,1): \Pr_{v_t,\ldots,v_k} 
    \left( \max_{i \in [k]} \dotp{{v_i}}{x} + (k-i)\gamma\|x\| \neq \max_{i \in [t]} \dotp{{v_i}}{x} + (k-i)\gamma\|x\| \right) \leq \frac{1}{n^{7}}.
\end{equation}    
\end{lemma}
\begin{proof}
    Let $E_x$ denote the event whose probability we want to upper bound. Since $E_x$ and $E_{\alpha x}$, for any $\alpha \in [0,1]$, are the same event, we can assume without loss of generality that $\norm{x}=1$.
        If event $E_x$ occurs, then it must hold that
    \begin{equation}
        \max_{i \in \{t+1,\ldots,k\}} \dotp{{v_i}}{x} + (k-i)\gamma > 
        \max_{i \in [t]} \dotp{{v_i}}{x} + (k-i)\gamma \geq 
        \dotp{{v_t}}{x} + (k-t)\gamma. \label{eq:Ex}
    \end{equation} 
    We want to show that this event is very unlikely. To do so, let $F_x$ be the event that for all $i \in \{t,\ldots,k\}$, $\dotp{{v_i}}{x} \in [-\frac{\gamma}{2},+\frac{\gamma}{2}]$. Note that if $F_x$ occurs, then the terms in the max are in decreasing order, and we have 
    \begin{equation}
        \dotp{{v_t}}{x} + (k-t)\gamma \geq \dotp{{v_{t+1}}}{x} + (k-t-1)\gamma \geq \cdots \geq \dotp{{v_{k-1}}}{x} + \gamma  \geq \dotp{{v_k}}{x},
    \end{equation}
    which contradicts \cref{eq:Ex}. Thus if $E_x$ holds then the complement of $F_x$, $\bar{F_x}$ must hold, which means $\Pr(E_x) \leq \Pr(\bar{F_x})$. So let us show that $F_x$ is very likely.
    
    The event $\bar{F_x}$ holds only if there exists an $i \in \{t,\ldots,k\}$ such that $\dotp{{v_i}}{x} \notin [-\frac{\gamma}{2},+\frac{\gamma}{2}]$. We can upper bound this probability for any particular $i \in \{t,\ldots,k\}$ using \Cref{prop:concentration} and the fact that $v_i$ is chosen uniformly at random from an $n-t+1$-dimension ball. This probability is at most $2e^{-(n-t+1)\gamma^2/8} = 2e^{-(n-t+1) \cdot 8 \frac{\log n}{n}} \leq 2 \cdot 2^{-8 \log n}= 2/n^{8}$, with the inequality holding because $n > 4t$. The probability that this happens for any $i$ is at most $(k-t+1) \leq k$ times this probability, by the union bound. Using the fact that $2k < n$, we get that $\Pr(E_x) \leq \Pr(\bar{F_x}) < 1/n^{7}$.
\end{proof}
Finally, we show that even if we knew the vectors $v_1,\ldots,v_{k-1}$, we cannot guess a vector $x$ that is an $\eps$-approximate solution to our problem, because it won't have enough overlap with $v_k$, which is unknown. In other words, for an algorithm to output an $\eps$-optimal solution, it essentially must know the entire set $V$.

\begin{lemma}[Cannot guess $x^*$]\label{lem:guess}
    Let $k>0$ be an integer and $\{v_1,\ldots,v_{k-1}\}$ be a set of orthonormal vectors. Let $v_k$ be chosen uniformly at random from $\spn(v_1,\ldots,v_{k-1})^{\perp}$ and let $V=(v_1,\ldots,v_k)$. Then
\begin{equation}
    \forall x \in B(\vec 0,1): \Pr_{v_k} 
    \left(f_V(x) - f_V(x^*) \leq \eps \right) \leq 2e^{-\Omega(k^2)}.
\end{equation}    
\end{lemma}
\begin{proof}
From \Cref{lem:propmin}, we know that an $\eps$-optimal solution $x$ must satisfy $\<v_k,x\> \leq -8\eps$. But $v_k$ is chosen uniformly at random from the space $\spn(v_1,\ldots,v_{k-1})^{\perp}$ and any vector $x \in B(\vec 0 , 1)$ projected to that space also has length at most $1$. So from \Cref{prop:concentration} we know that for any $x \in B(\vec 0 , 1)$,
\begin{equation}
    \Pr_{v_k} (\<v_k,x\> \leq -8\eps) \leq \Pr_{v_k} (|\<v_k,x\>| \geq 8\eps) \leq 2e^{-32(n-k+1)\eps^2} \leq 2e^{-\Omega(k^2)}. \qedhere
\end{equation}
\end{proof}

\subsection{Quantum query model} 

We now formally define the quantum query model in our setting. In the usual quantum query model the set of allowed queries is finite, whereas in our setting it is natural to allow the quantum algorithm to query the oracles at any point $x \in \R^n$. Due to \Cref{lem:propf}, it is sufficient to allow the algorithm to query any $x \in B(\vec 0,1)$, but this is still a continuous space of queries, and hence a query vector could be a superposition over infinitely many states. Instead of formalizing this notion of quantum algorithms, we allow the algorithm to make discrete queries only, but to arbitrarily high precision. The reader is encouraged to not get bogged down by details and to think of the registers as storing the real values that they ideally should, but in the rest of this section we define these algorithms more carefully so that all the spaces involved are finite and well defined. This formalization is not specific to the quantum setting and is done classically as well if we do not want to manipulate real numbers as atomic objects.

All the real numbers that appear will be represented using some $b$ bits of precision, where $b$ can be chosen by the algorithm. The reader should imagine $b$ being arbitrarily large, say exponentially larger than all the parameters involved in the problem, so that the inaccuracy involved by using this representation is negligible. Then the algorithm represents the input $x \in B(\vec 0,1)$ using $b$ bits of precision per coordinate. The oracle's response will also use $b$ bits of precision per real number. For a given choice of $b$, the quantum algorithm will have some probability of success of solving the problem at hand. We then define the success probability of quantum algorithms that make $q$ queries by taking a supremum over all $b$ of $q$-query algorithms that solve the problem.

We can now define the oracles more precisely. Classically, the function oracle for a function $f:\R^n \to \R$ would simply implement the map $x \mapsto f(x)$, where we represent each entry of $x$ and the output $f(x)$ using $b$ bits, so $x\in \B^{bn}$ and $f(x)\in \B^b$. Let's say we have a classical circuit that implements this map using $G$ gates, say over the gate set of AND, OR, and NOT gates. Then it is easy to construct, in a completely black-box way, a quantum circuit using $O(G)$ gates (say over the gate set of Hadamard, CNOT, and T) that performs the unitary $U|x\>|y\> = |x\>|y \oplus f(x)\>$, for every $x\in \B^{bn}$, and $y \in \B^b$. This is why it is standard to assume that the quantum oracle corresponding to the classical map $x \mapsto f(x)$ is a unitary that performs $U|x\>|y\> = |x\>|y \oplus f(x)\>$. We apply the same construction for the $\FO(f)$ oracle to get the quantum analogue of the classical map $x \mapsto g_x$, where $g_x$ is some subgradient of $f$ at $x$. Lastly, for convenience we will combine both the function and subgradient oracle into one oracle that when queried with $x$ returns $f(x)$ and a subgradient at $x$.  Since our function family is parameterized by $V \in \mV$, we call this oracle $O_V$.

Let $\mA$ be a quantum query algorithm that makes $q$ queries.  $\mA$ is described by a sequence of unitaries $U_{q}O_{V}U_{q-1}O_{V}U_{q-2}O_{V} \cdots U_{1}O_{V}U_{0}$ applied to an initial state, say $\ket{0}$. 
We assume that the output of $\mA$, which is a vector $x$, is determined by measuring the first $n$ registers storing real numbers using $b$ bits.

\subsection{Lower bound}
\label{subsec:quantumlowerbound}

We can now prove the quantum lower bound. Let $\mA$ be a $k-1$ query quantum algorithm that solves \Cref{prob:FOCM} on all the functions $f_V$ for $V \in \mV$. 
Due to \Cref{lem:propf}, we can assume that the algorithm only queries the oracles with vectors $x \in B(\vec 0,1)$.
We also need to also describe the behavior of the subgradient oracle on inputs where the subgradient is not unique. 
On such inputs $x$, the subgradient is not unique because several indices $i \in [k]$ simultaneously achieve the maximum in $f_V(x)$. 
In this case, the subgradient will answer as if the smallest index $i$ in this set achieved the maximum. Now let $\mA$ be described by the sequence of unitaries $U_{k-1}O_{V}U_{k-2}O_{V} \cdots O_VU_{1}O_{V}U_{0}$ acting on the starting state $|0\>$. Let this sequence of unitaries be called $A$. Then the final state of the algorithm is $A|0\>$.  

Recall that we defined $f_V(x) = \max_{i \in [k]} \{g_V^{(i)}(x)\}$. Let us also define functions $f_V^{(j)}$ where the maximization is only over the first $j$ indices instead of all $k$ indices. Specifically, let $f_V^{(j)} \defeq  \max_{i \in [j]} \{g_V^{(i)}(x)\}$. We previously defined the oracle $O_V$ as corresponding to the function $f_V$. Let $O_V^{(j)}$ be the oracle corresponding to the functions $f_V^{(j)}$. 

Now we define a sequence of unitaries starting with $A_0 = A$ as follows:
\begin{align}\label{eq:unitaries}
    A_0 &\defeq U_{k-1}O_{V}U_{k-2}O_{V} \cdots O_VU_{1}O_{V}U_{0} \nonumber\\
    A_1 &\defeq U_{k-1}O_{V}U_{k-2}O_{V} \cdots O_VU_{1}O^{(1)}_{V}U_{0} \nonumber \\
    A_2 &\defeq U_{k-1}O_{V}U_{k-2}O_{V} \cdots O_V^{(2)}U_{1}O^{(1)}_{V}U_{0}\\
    &\phantom{n}\vdots \nonumber\\
    A_{k-1} &\defeq U_{k-1}O_{V}^{(k-1)}U_{k-2}O_{V}^{(k-2)} \cdots O_V^{(2)}U_{1}O^{(1)}_{V}U_{0} \nonumber
\end{align}

We want to show that the algorithm $A_0$ does not solve our problem. To do so, we will employ the hybrid argument, in which we show that the output of the algorithm $A_i$ and $A_{i+1}$ is close, and thus the output of $A_0$ and $A_{k-1}$ is close. Finally, we argue that the algorithm $A_{k-1}$ does not solve our problem because the oracles in the algorithm do not know $v_k$. Let us first establish these two claims.

\begin{lemma}[$A_{k-1}$ does not solve the problem]\label{lem:Akminusone}
    Let $\mA$ be a $k-1$ query algorithm and let $A_{k-1}$ be defined as above. 
    Let $p_V$ be the probability distribution over $x \in B(\vec 0,1)$ obtained by measuring the output state $A_{k-1}|0\>$. Then $\Pr_{V, x\sim p_V} (f_V(x)-f_V(x^*)\leq \eps) \leq 2e^{-\Omega(k^2)}$.
\end{lemma}
\begin{proof}
We want to show that the probability (over the random choice of $V$ and the internal randomness of the algorithm) that $A_{k-1}$ outputs an $x$ that satisfies $f(x)-f(x^*)\leq \eps$ is very small. 

Let us establish the claim for any fixed choice of $v_1,\ldots v_{k-1}$, since if the claim holds for any fixed choice of these vectors, then it also holds for any probability distribution over them. For a fixed choice of vectors, this claim is just  $\Pr_{v_k, x \sim p_V} (f_V(x)-f_V(x^*)\leq \eps) \leq 2e^{-\Omega(k^2)}$. Now since the algorithm $A_{k-1}$ only has oracles $O_V^{(i)}$ for $i<k$, the probability distribution $p_V$ only depends on $v_1,\ldots,v_{k-1}$. Since these are fixed, this is just a fixed distribution $p$. So we can instead establish our claim for all $x \in B(\vec 0, 1)$, which will also establish it for any distribution.

So what we need to establish is that for any $x \in B(\vec 0,1)$, $\Pr_{v_k} (f_V(x)-f_V(x^*)\leq \eps) \leq 2e^{-\Omega(k^2)}$, which is exactly what we showed in \Cref{lem:guess}.
\end{proof}

\begin{lemma}[$A_{t}$ and $A_{t-1}$ have similar outputs]\label{lem:similar}
    Let $\mA$ be a $k-1$ query algorithm and let $A_{t}$ for $t\in [k-1]$ be the unitaries defined in \cref{eq:unitaries}. Then
    \begin{equation}
        \E_V\bigl(\norm{A_t|0\>-A_{t-1}|0\>}^2\bigr) \leq \frac{4}{n^{7}}.
    \end{equation}
\end{lemma}
\begin{proof}
    From the definition of the unitaries in \cref{eq:unitaries} and the unitary invariance of the spectral norm, we see that $\norm{A_t|0\>-A_{t-1}|0\>} = 
    \norm{ (O_V^{(t)}-O_V) U_{t-1}O_V^{(t-1)} \cdots O_V^{(1)} U_0|0\>}$. Let us again prove the claim for any fixed choice of vectors $v_1,\ldots,v_{t-1}$, which will imply the claim for any distribution over those vectors. Once we have fixed these vectors, the state $U_{t-1}O_V^{(t-1)} \cdots O_V^{(1)} U_0|0\>$ is a fixed state, which we can call $|\psi\>$. 
    Thus our problem reduces to showing for all quantum states $|\psi\>$,
    \begin{equation}\label{eq:OVOVt}
        \E_{v_t,\ldots,v_k}\bigl(\norm{(O_V^{(t)}-O_V)|\psi\>}^2\bigr) \leq \frac{4}{n^{7}}. 
    \end{equation}
    Now we can write an arbitrary quantum state as $|\psi\>=\sum_x \alpha_x |x\>|\phi_x\>$, where $x$ is the query made to the oracle, and $\sum_x |\alpha_x|^2 =1$.  Thus the LHS of \cref{eq:OVOVt} is equal to
    \begin{equation}
        \E_{v_t,\ldots,v_k}\left(\Norm{\sum_{x} \alpha_x (O_V^{(t)}-O_V)|x\>|\phi_x\>}^2\right)
        \leq 
        \sum_{x}  |\alpha_x|^2 \E_{v_t,\ldots,v_k}\left(\norm{  (O_V^{(t)}-O_V)|x\>}^2\right)        .
    \end{equation}

    Since $|\alpha_x|^2$ defines a probability distribution over $x$, we can again upper bound the right hand side for any $x$ instead. 
    Since $O_V^{(t)}$ and $O_V$ behave identically for some inputs $x$, the only nonzero terms are those where the oracles respond differently, which can only happen if $f_V^{(t)}(x) \neq f_V(x)$. When the response is different, we can upper bound $\norm{  (O_V^{(t)}-O_V)|x\>}^2$ by $4$ using the triangle inequality. Thus for any $x \in B(\vec 0,1)$, we have 
    \begin{equation}
        \E_{v_t,\ldots,v_k}\left(\norm{  (O_V^{(t)}-O_V)|x\>}^2\right)        
        \leq 4 \Pr_{v_t,\ldots,v_k}(f_V^{(t)}(x) \neq f_V(x)) \leq 4/n^{7}, 
    \end{equation}
    where the last inequality follows from \Cref{lem:argmax}.
\end{proof}

Finally we can put these two lemmas together to prove our lower bound.

\begin{lemma}[$\mA$ does not solve the problem]\label{lem:qlowerbound}
    Let $\mA$ be a $k-1$ query algorithm. Let $p_V$ be the probability distribution over $x \in B(\vec 0,1)$ obtained by measuring the output state $A|0\>$. Then $\Pr_{V, x\sim p_V} (f_V(x)-f_V(x^*)\leq \eps) \leq \frac{1}{\mathrm{poly}(n)}$.
\end{lemma}

\begin{proof}
    Let $P_V$ be the projection operator that projects a quantum state $\ket{\psi}$ onto the space spanned by vectors $\ket{x}$ for $x$ such that $f_V(x)-f_V(x^*)\leq \eps$. Then $\| P_V A \ket{0} \|^2 = \Pr_{x\sim p_V} (f_V(x)-f_V(x^*)\leq \eps)$. We know from \Cref{lem:Akminusone} that $\E_V\bigl(\norm{P_V A_{k-1} \ket{0}}^2\bigr) \leq 2e^{-\Omega(k^2)}$. We prove our upper bound on the probability by showing that it is approximately the same as $\E_V\bigl(\norm{P_V A_{k-1} \ket{0}}^2\bigr)$.
    
    \Cref{lem:similar} states that for all $1 \leq t < k$, $\E_V\bigl(\norm{A_{t}|0\>-A_{t-1}|0\>}^2\bigr) \leq \frac{4}{n^{7}}$. Using telescoping sums and the Cauchy-Schwarz inequality, we see that
    \begin{align}
        \E_V\bigl( \norm{A_{k-1}|0\>-A|0\>}^2 \bigr) &\leq \E_V\left( \left(\sum_{t \in [k-1]} \norm{A_{t}|0\>-A_{t-1}|0\>}\right)^2 \right)\\
        &\leq \E_V\left(\sum_{t \in [k-1]} \norm{A_{t}|0\>-A_{t-1}|0\>}^2\right) \left( \sum_{t \in [k-1]} 1^2 \right) \leq \frac{4k}{n^{7}} \cdot k.
    \end{align}

    For all $V$, $\abs{\norm{P_V A_{k-1} \ket{0}} - \norm{P_V A \ket{0}}} \leq \norm{P_V A_{k-1} \ket{0} - P_V A \ket{0}} = \norm{P_V (A_{k-1} \ket{0} - A \ket{0})}  \leq \norm{A_{k-1} \ket{0} - A \ket{0}}$.

    Hence $\E_V\bigl(\bigl(\norm{P_V A_{k-1} \ket{0}} - \norm{P_V A \ket{0}}\bigr)^2\bigr) \leq \frac{4k^2}{n^{7}}$. By Markov's inequality, $\Pr_V\bigl(\bigl(\norm{P_V A_{k-1} \ket{0}} - \norm{P_V A \ket{0}}\bigr)^2 \geq \frac{1}{n^4} \bigr) \leq \frac{4k^2}{n^{3}}$. So it is overwhelmingly likely that $\norm{P_V A \ket{0}} - \norm{P_V A_{k-1} \ket{0}} \leq \frac{1}{n^2}$, which implies $\norm{P_V A \ket{0}}^2 - \norm{P_V A_{k-1} \ket{0}}^2 \leq \frac{2}{n^2}$ since both norms are at most $1$. Even assuming that in the unlikely cases the difference is the maximum possible, we still get $\E_V\bigl(\norm{P_V A \ket{0}}^2 - \norm{P_V A_{k-1} \ket{0}}^2\bigr) \leq \frac{4k^2}{n^{3}} + \frac{2}{n^2}$.
    
    We can now use linearity of expectation and upper bound our required probability as 
    \begin{equation} 
        \Pr_{V,x\sim p_V} (f_V(x)-f_V(x^*)\leq \eps) = \E_V\bigl(\norm{P_V A \ket{0}}^2\bigr) \leq 2e^{-\Omega(k^2)} + \frac{4k^2}{n^{3}} + \frac{2}{n^2}. \qedhere
    \end{equation}
\end{proof}

Note that this establishes a statement similar to \Cref{thm:quantumlb}, except with a polynomially larger value of $n$. This result is sufficient to establish the optimality of gradient descent in the dimension-independent setting. In the next section we quantitatively improve the lower bound by reducing the value of $n$.

\subsection{Improved lower bound using the wall function} \label{sec:wall}
\renewcommand{\nem}{p}
\newcommand{\newf}{f}

In this section we improve the dimension dependence of the previous lower bound using the strategy used by \cite{BJLLS19}, where they introduce a function called the \emph{wall function}. We now provide a high-level overview of this strategy before getting into the details.

The previous construction required a larger dimension $n$ because we needed to use a large value of $\gamma$, which in turn was large because we wanted the following key property (i.e., \Cref{lem:argmax}) to hold: 
If you query the function $f_V(x)$ with a random vector $x \in \R^n$, the function is almost certainly maximized on the first term in the max, and the answer of the gradient oracle is $v_1$. 
To reduce the parameter $n$, we will use a different function in this section. This function will be built out of the functions $p_V:\R^n \to \R$, where $V=(v_1,v_2,\ldots,v_k)$ is again a set of $k$ orthonormal vectors:
\begin{equation} \label{eq:pV}
    \nem_V(x) \defeq \max_{i \in [k]} \{\dotp{v_i}{x} - i\gamma\}, 
\end{equation}
where $\gamma$ is unspecified for now. If we only allow the algorithm to query the oracle with a vector $x$ with $\norm{x}=1$, this function is essentially the same as the function $f_V$ we used in the previous section, up to an additive $k\gamma$ term. Allowing the algorithm to query $p_V$ at vectors $x$ with $\norm{x} \leq 1$ is fine too, since our key property will still hold: Querying the gradient oracle with a random $x$ with norm less than $1$ will still return $v_1$ almost certainly. But if we allow the algorithm to query with vectors $x$ with extremely large norm, the additive term $i\gamma$ will be negligible, and the property we want (that the answer is almost certainly $v_1$) will not hold anymore. 

The wall function construction is a way of fixing this problem. The wall function constrains the set of points that can be queried to gain useful information about the set $V$. At the beginning, when the algorithm does not know the set $V$, the wall function essentially forces the algorithm to query the oracle with vectors $x$ with small norm. If the oracle is queried with a vector of large norm, the wall function ``hides'' information about the set $V$ by outputting an answer that (with high probability) is independent of $V$. More generally, if the algorithm has learned a subset of $V$, and the algorithm queries the oracle with a vector $x$ with a large projection outside of the span of the vectors it knows, then (with high probability) the oracle's answer hides information about $V$. In this setting, querying a unit vector at random would be inadvisable since the whole vector would be outside of the span of the vectors the algorithm knows, and the oracle's response will be non-informative. The useful queries will be shorter vectors which do not trigger the wall function's obfuscation, since any projection of a short vector is also short. This restriction on the query vector length now allows us to choose a smaller value of $\gamma$ than in the previous construction, and hence have a smaller dimension $n$. 

\para{Formal construction.}
We now describe the construction formally. As in the previous section, $V=(v_1,\ldots,v_k)$ is a set of $k$ orthonormal vectors in $\R^n$. Our family of functions will depend on several parameters ($n$, $k$, $\delta$, and $\gamma$), which are all functions of $\eps$, which is the single parameter on which the function family depends. 

Let us start with $k$. As before, we will show a lower bound of $\Omega(k)$, and so we want $k$ to be a small multiple of $1/\eps^2$. Thus we choose  $k \defeq \frac{1}{100\eps^2}$. For some large enough constant $c$, we set
\begin{equation}\label{eq:nk}
    n \defeq c k^2 \log k = \tO\left(\frac{1}{\epsilon^4}\right).
\end{equation}
This is chosen to satisfy \cref{eqn:wallkgamma}. 
Let $\delta$ be chosen such that 
\begin{equation} \label{eq:delta}
    \frac{\delta}{\log(1/\delta)} \defeq 32 \sqrt{\frac{k \log n}{n}} + \frac{1}{\sqrt{k}} = \Theta\left(\frac{1}{\sqrt{k}}\right). 
\end{equation}
This value is chosen to make the first property in \Cref{lem:aboutthewall} hold. Let $p_V$ be the function defined in \cref{eq:pV} with $\gamma$ defined as
\begin{equation}\label{eqn:wallgamma}
    \gamma \defeq 8 \delta \sqrt{\log n/n}.
\end{equation}
This value is chosen for a similar reason to before, and more precisely it is required in \Cref{lem:wallargmax}. As in the previous lower bound (and for the same reason), we want
\begin{equation}\label{eqn:wallkgamma}
    k \gamma \leq \frac{1}{10\sqrt{k}}.
\end{equation}
Our choice of $n$ in \cref{eq:nk} satisfies \cref{eqn:wallkgamma}. 

Before constructing the wall function, we need to define the \emph{correlation cones} $C_1, \dots, C_k$, which depend on $v_1, \dots, v_k$:
\begin{equation} 
    C_i \defeq \left\{ x \in \R^n \middle| \frac{\abs{\dotp{v_i}{x}}}{\|x\|} \geq 8 \sqrt{\frac{\log n}{n}} \right\}. 
\end{equation}
Note that if you choose a random unit vector $x$, it will most likely not be in $C_i$ since the normalized inner product will be roughly $1/\sqrt{n}$. Thus $C_i$ is the set of directions that correlate strongly with $v_i$.

We define the set
\begin{equation}
    \Omega = \{ x \in \R^n \mid \|x\| \in [\delta,1] \wedge \forall i \in [k], ~ x \notin C_i \}    
\end{equation}
to be the set of vectors that have non-negligible norm and are not in any of the correlation cones $C_i$. We want our construction to give non-informative answers on $\Omega$ so that the algorithm is forced to query on the complement of $\Omega$. 

We now define our non-informative function $h(x) = 2\|x\|^{1 + \alpha}$, with $\alpha > 0$ set so that $\delta^\alpha = 1/2$. Note that because of the value of $\delta$ chosen, $1/\alpha=\Theta(\log n)$, so $\alpha$ is small.
We want our wall function to be equal to $h(x)$ on $\Omega$, but we need to define it everywhere in $\R^n$. We do so by extending this function to all of $\R^n$ by convexity. Informally this means that the function takes the smallest value it can outside $\Omega$ while remaining convex. Formally, we define the wall function as
\begin{equation} 
    \wall_V(x) \defeq \max_{y \in \Omega} \{ h(y) + \dotp{\nabla h(y)}{x - y} \}. 
\end{equation}

In the following lemma, we state some properties of the wall function established by \cite{BJLLS19}.

\begin{lemma}[Properties of the wall function]\label{lem:aboutthewall}
    The wall function satisfies the following properties.
\begin{enumerate}
    \item \cite[Lemma 2]{BJLLS19}: The point $\tilde{x} = -\sum_{i \in [k]} v_i/\sqrt{k}$ satisfies $\wall(\tilde{x}) \leq -1/\sqrt{k}$.
    \item \cite[Lemma 3]{BJLLS19}: Let $x \in \R^n$. For any $t \in [k]$, let $x = w+z$, where $w \in \spn(v_1,\dots,v_t)$ and $z \in \spn(v_1,\dots,v_t)^{\perp}$. If $\forall j>t, ~ z \notin C_j$, then $\wall_V(x)$ does not depend on $v_{t+1},\dots,v_k$. For such $x$, $\wall_V(x)$ takes the same value as the following function.
    \begin{equation} \wall_V^{(t)}(x) = \max_{a,b \in \R_{+} : a^2 + b^2 \in [\delta^2,1]} \left\{ -2 \alpha c^{1+\alpha} + 2 \frac{1+\alpha}{c^{1-\alpha}} \left( \max_{y \in \tilde{\Omega}_{a,b},\|y\|=a} \dotp{y}{w} + b\|z\| \right) \right\} \end{equation}
    where $c = \sqrt{a^2 + b^2}$ and $\tilde{\Omega}_{a,b} = \{ x \in \spn(v_1, \dots, v_t) \mid \forall i \leq t ~ \frac{\abs{\dotp{v_i}{x}}}{\|x\|} \frac{a}{\sqrt{a^2 + b^2}} < 8\sqrt{\log n/n} \}$.\footnote{Lemma 3 in~\cite{BJLLS19} gives a different definition of $\tilde{\Omega}$, but we believe this is the set they meant to define.}
    \item Discussion after \cite[Lemma 3]{BJLLS19}: Furthermore, if $\forall j > t, ~ z \notin C_j$ and $\|z\| \geq \delta$, then $\max_{i \in [k]} \dotp{v_i}{x} - i\gamma \neq \max_{i \in [t]} \dotp{v_i}{x} - i\gamma$ implies that $\wall_V(x) \geq \nem_V(x)$.
\end{enumerate}
\end{lemma}

The second property in the lemma implies that the value of the wall function on $x$ grows with $z$, the uncorrelated projection of $x$. The third item states that if $z$ is somewhat large and the maximum in the definition of $\nem_V$ is achieved at an index with $i>t$, then $\wall_V(x)$ is actually larger than $\nem_V(x)$.

Equipped with these properties, we define the actual class of functions. For a set $V$ of $k$ orthonormal vectors, we define 
\begin{equation}
    \newf_V(x) \defeq \max \{ \nem_V(x), \wall_V(x) \}.    
\end{equation}
Note that $\wall_V(x)$ is also convex, being a maximum of linear functions. The maximum norm of the gradient of any of the linear functions is $2(1+\alpha)$ and hence the function is $3$-Lipschitz. Each of the $\nem_V$ functions is $1$-Lipschitz. Hence $\newf_V$ is also $3$-Lipschitz.
This completely specifies the function $\newf_V$ that we will use for a given value of $\epsilon$. 

We also define the following functions.
\begin{equation} \nem_V^{(t)}(x) \defeq \max_{i \in [t]} \{\dotp{v_i}{x} - i\gamma\} \text{ and } \newf_V^{(t)}(x) = \max \{ \nem_V^{(t)}(x), \wall_V^{(t)}(x) \}. \end{equation}
Note that if the preconditions in item 2 of \Cref{lem:aboutthewall} are satisfied for some value $t$, then they are also satisfied for $t+1$. So for such $x$, $\wall_V(x) = \wall_V^{(t)}(x) = \wall_V^{(t+1)}(x)$.

\Cref{lem:aboutthewall} implies some very convenient statements. Let $v_1, \dots, v_k$ be fixed orthonormal vectors. Then for any point $x = w+z$, where $w \in \spn(v_1,\dots,v_{t-1})$ and $z \in \spn(v_1,\dots,v_{t-1})^{\perp}$ we can make the following statements.

\begin{itemize}
    \item If $\|z\| \geq \delta$ and $\forall j \geq t, ~ z \notin C_j$,\\
        then $\wall_V(x) = \wall_V^{(t-1)}(x)$ and also $\nem_V(x) \neq \nem_V^{(t-1)}(x) \implies \wall_V(x) \geq \nem_V(x)$. Hence $\newf_V(x) = \newf_V^{(t-1)}(x)$.
    \item If $\|z\| < \delta$ and $\forall j \geq t, ~ z \notin C_j$ and $\nem_V(x) = \nem_V^{(t)}(x)$,\\
        then $\wall_V(x) = \wall_V^{(t-1)}(x)$ and $\nem_V(x) = \nem_V^{(t)}(x)$. So $\newf_V(x) = \newf_V^{(t)}(x)$.
\end{itemize}

We now show the following lemma, akin to \Cref{lem:argmax}.

\begin{lemma}\label{lem:wallargmax}
    Let $1\leq t \leq k$ be integers and $\{v_1,\ldots,v_{t-1}\}$ be a set of orthonormal vectors. Let $\{v_t,\ldots,v_k\}$ be chosen uniformly at random so that the set $\{v_1,\ldots,v_k\}$ is orthonormal. Then 
\begin{equation}\label{eqn:exagain}
    \forall x \in \R^n: \Pr_{v_t,\ldots,v_k} 
    \left( \newf_V(x) \neq \newf_V^{(t)}(x) \right) \leq \frac{1}{n^{7}}.
\end{equation}    
\end{lemma}
\begin{proof}
    Let $x = w+z$, where $w \in \spn(v_1,\dots,v_{t-1})$ and $z \in \spn(v_1,\dots,v_{t-1})^{\perp}$.
    Let $E_x$ denote the event whose probability we want to upper bound.
    
    If $\|z\| \geq \delta$, then $E_x$ can only occur if $z \in C_j$ for some $j \geq t$. Using \Cref{prop:concentration}, the fact that each $v_j$ in this range is chosen uniformly at random from an $n-t+1$-dimensional ball and a union bound, this probability is upper bounded by $k \cdot 2e^{-(n-t+1) \cdot 32 \frac{\log n}{n}} \leq 2k \cdot 2^{-32\log n} = 2k/n^{32} \leq 1/n^{31}$, with the inequalities holding because $n > 4k$.

    If $\|z\| \leq \delta$, then $E_x$ can only occur if $z \in C_j$ for some $j \geq t$ \emph{or} $\nem_V(x) \neq \nem_V^{(t)}(x)$. The former probability we have already upper bounded by $1/n^{31}$. The latter probability can be upper bounded as follows. If the latter event, let us call it $E'_x$, occurs then it must hold that
    \begin{equation}
        \max_{i \in \{t+1,\ldots,k\}} \dotp{{v_i}}{x} - i\gamma = \max_{i \in \{t+1,\ldots,k\}} \dotp{{v_i}}{z} - i\gamma > 
        \max_{i \in [t]} \dotp{{v_i}}{x} - i\gamma \geq 
        \dotp{{v_t}}{z} - t\gamma.
    \end{equation} 
    We will show that this event is very unlikely. To do so, let $F_x$ be the event that for all $i \in \{t,\ldots,k\}$, $\dotp{{v_i}}{z} \in [-\frac{\gamma}{2},+\frac{\gamma}{2}]$. Note that if $F_x$ occurs, then the terms in the max are in decreasing order, and we have 
    \begin{equation}
        \dotp{{v_t}}{z} - t\gamma \geq \dotp{{v_{t+1}}}{z} - (t+1)\gamma \geq \cdots \geq \dotp{{v_{k-1}}}{z} - (k-1)\gamma \geq \dotp{{v_k}}{z} - k\gamma,
    \end{equation}
    which contradicts the previous equation. Thus if $E'_x$ holds then the complement of $F_x$, $\bar{F_x}$ must hold, which means $\Pr(E'_x) \leq \Pr(\bar{F_x})$. So let us show that $F_x$ is very likely.
    
    The event $\bar{F_x}$ holds if for any $i \in \{t,\ldots,k\}$, $\dotp{{v_i}}{z} \notin [-\frac{\gamma}{2},+\frac{\gamma}{2}]$. We can upper bound this probability for any particular $i \in \{t,\ldots,k\}$ using \Cref{prop:concentration} and the fact that $v_i$ is chosen uniformly at random from an $n-t+1$-dimension ball. This is the same as the probability that $\dotp{{v_i}}{z/\delta} \notin [-4\sqrt{\log n/n},4\sqrt{\log n/n}]$. Since $z/\delta \in B(\vec 0,1)$, this is at most $2e^{-(n-t+1) \cdot 8 \frac{\log n}{n}} \leq 2 \cdot 2^{-8\log n} = 2/n^{8}$,  with the inequality holding because $n > 4t$. The probability that this happens for any $i$ is at most $k$ times this probability, by the union bound. Using the fact that $4k < n$, we get that $\Pr(E'_x) \leq \Pr(\bar{F_x}) < 1/2n^{7}$.

    Putting it all together, we can upper bound the probability in the lemma statement by the maximum of $1/n^{31}$ and $1/n^{31} + 1/2n^7$, and so the lemma follows.
\end{proof}

Letting $\tilde{x} = - \sum_{i \in [k]} v_i/\sqrt{k}$, it is clear that $\nem_V(\tilde{x}) \leq -1/\sqrt{k}$. We have also seen that $\wall_V(\tilde{x}) \leq -1/\sqrt{k}$. So $\newf_V(\tilde{x}) \leq -1/\sqrt{k} = -10\epsilon$. Any point $x$ minimizing $\newf_V$ to within $\epsilon$ of the optimum must satisfy $\nem_V(x) \leq -9\epsilon$. From \cref{eqn:wallkgamma}, we see that $k \gamma \leq \epsilon$. So the point $x$ must also satisfy $\dotp{v_k}{x} \leq -8\epsilon$. Using this, we get the following analog of \Cref{lem:guess}, whose proof is identical.

\begin{lemma}\label{lem:guessagain}
    Let $k>0$ be an integer and $\{v_1,\ldots,v_{k-1}\}$ be a set of orthonormal vectors. Let $v_k$ be chosen uniformly at random from $\spn(v_1,\ldots,v_{k-1})^{\perp}$ and let $V=(v_1,\ldots,v_k)$. Then
\begin{equation}
    \forall x \in B(\vec 0,1): \Pr_{v_k} 
    \left(f_V(x) - f_V(x^*) \leq \eps \right) \leq 2e^{-\Omega(k)}.
\end{equation}
\end{lemma}

Finally, since the lemmas in the proof of the previous section's quantum lower bound (\Cref{subsec:quantumlowerbound}) used these two lemmas as a black box, the same proof allows us to argue that no algorithm can perform well if it makes at most $k-1$ queries to the oracle. This completes the proof of \Cref{thm:quantumlb}.  
\section{Open problems}
\label{sec:open}

We showed that in the black-box setting, no quantum algorithm can beat gradient descent in general, in the dimension-independent regime. Here are some interesting questions left open by our work:

\begin{enumerate}
  \item We showed in \Cref{thm:quantumub} that the class of functions used in the randomized lower bound can be solved faster with quantum queries. Is there a more interesting class of functions on which we can achieve a quantum speedup?
  \item Can the quantum lower bound in \Cref{sec:quantumlb} be made to work using the simpler class of functions $f_V(x) = \max_i \<v_i,x\>$, which is our function with $\gamma = 0$? If so, this might also decrease the dimension $n$ required.
  \item Can we establish tight quantum lower bounds in the parameter regime where dimension-dependent algorithms outperform gradient descent? When $1/\eps$ is a large polynomial in $n$, the complexity of gradient descent is also a large polynomial in $n$, but a dimension-dependent algorithm such as the center of gravity method~\cite{Bub15} yields an $O(n \log n)$ upper bound. Can we establish an $\tOmega(n)$ lower bound in this regime? The function used in the randomized lower bound in \Cref{sec:randomizedlb} yields an $\tOmega(\sqrt{n})$ lower bound and this is the best bound we are aware of. This is essentially the same as the problem left open by \cite{CCLW20,vAGGdW20}, but phrased in the language of membership and separation oracles.
  \item What can we say about other standard settings in convex optimization beyond first-order non-smooth convex optimization? Other natural settings include assuming the function is smooth, having the ability to query a prox oracle instead of a subgradient oracle, etc. Can quantum algorithms provide a speedup in the black-box model in these settings over the respective best classical algorithms in that setting?
\end{enumerate}

\section*{Acknowledgements}

We thank S\'{e}bastien Bubeck, Ronald de Wolf, and Andr{\'a}s Gily{\'{e}}n for helpful conversations about this work. 
RK thanks Vamsi Pritham Pingali for many helpful conversations about multivariable calculus.

\bibliographystyle{alphaurl}
\bibliography{references}

\newcommand{\etalchar}[1]{$^{#1}$}
\begin{thebibliography}{{\noop{Apeldoorn}v}AGGdW20}

\bibitem[{\noop{Apeldoorn}v}AG19]{vAG19}
Joran {\noop{Apeldoorn}v}an~Apeldoorn and Andr{\'a}s Gily{\'e}n.
\newblock {Improvements in Quantum SDP-Solving with Applications}.
\newblock In {\em 46th International Colloquium on Automata, Languages, and
  Programming (ICALP 2019)}, volume 132 of {\em Leibniz International
  Proceedings in Informatics (LIPIcs)}, pages 99:1--99:15. Schloss
  Dagstuhl--Leibniz-Zentrum fuer Informatik, 2019.
\newblock \href {https://doi.org/10.4230/LIPIcs.ICALP.2019.99}
  {\path{doi:10.4230/LIPIcs.ICALP.2019.99}}.

\bibitem[{\noop{Apeldoorn}v}AGGdW17]{vAGGdW17}
Joran {\noop{Apeldoorn}v}an~Apeldoorn, Andr{\'{a}}s Gily{\'{e}}n, Sander
  Gribling, and Ronald de~Wolf.
\newblock Quantum {SDP}-solvers: Better upper and lower bounds.
\newblock In {\em 58th Annual Symposium on Foundations of Computer Science
  ({FOCS 2017})}, oct 2017.
\newblock \href {https://doi.org/10.1109/focs.2017.44}
  {\path{doi:10.1109/focs.2017.44}}.

\bibitem[{\noop{Apeldoorn}v}AGGdW20]{vAGGdW20}
Joran {\noop{Apeldoorn}v}an~Apeldoorn, Andr{\'{a}}s Gily{\'{e}}n, Sander
  Gribling, and Ronald de~Wolf.
\newblock Convex optimization using quantum oracles.
\newblock {\em {Quantum}}, 4:220, January 2020.
\newblock \href {https://doi.org/10.22331/q-2020-01-13-220}
  {\path{doi:10.22331/q-2020-01-13-220}}.

\bibitem[Bal97]{Ball97}
Keith Ball.
\newblock An elementary introduction to modern convex geometry.
\newblock In Silvio Levy, editor, {\em Flavors of geometry}, volume~31, pages
  1--58. Cambridge University Press, 1997.
\newblock URL: \url{http://library.msri.org/books/Book31/files/ball.pdf}.

\bibitem[BBBV97]{BBBV97}
Charles~H. Bennett, Ethan Bernstein, Gilles Brassard, and Umesh Vazirani.
\newblock Strengths and weaknesses of quantum computing.
\newblock {\em SIAM Journal on Computing}, 26(5):1510--1523, 1997.
\newblock \href {https://doi.org/10.1137/S0097539796300933}
  {\path{doi:10.1137/S0097539796300933}}.

\bibitem[Bel14]{Bel14}
Aleksandrs Belovs.
\newblock Quantum algorithms for learning symmetric juntas via adversary bound.
\newblock In {\em Proceedings of the 2014 IEEE 29th Conference on Computational
  Complexity (CCC 2014)}, CCC ’14, page 22–31, 2014.
\newblock \href {https://doi.org/10.1109/CCC.2014.11}
  {\path{doi:10.1109/CCC.2014.11}}.

\bibitem[BJL{\etalchar{+}}19]{BJLLS19}
S{\'{e}}bastien Bubeck, Qijia Jiang, Yin~Tat Lee, Yuanzhi Li, and Aaron
  Sidford.
\newblock Complexity of highly parallel non-smooth convex optimization.
\newblock In {\em Advances in Neural Information Processing Systems 32 (NeurIPS
  2019)}, pages 13900--13909, 2019.
\newblock URL:
  \url{http://papers.nips.cc/paper/9541-complexity-of-highly-parallel-non-smooth-convex-optimization}.

\bibitem[BKL{\etalchar{+}}19]{BKL+19}
Fernando G. S.~L. Brand{\~a}o, Amir Kalev, Tongyang Li, Cedric Yen-Yu Lin,
  Krysta~M. Svore, and Xiaodi Wu.
\newblock {Quantum SDP Solvers: Large Speed-Ups, Optimality, and Applications
  to Quantum Learning}.
\newblock In {\em 46th International Colloquium on Automata, Languages, and
  Programming (ICALP 2019)}, volume 132 of {\em Leibniz International
  Proceedings in Informatics (LIPIcs)}, pages 27:1--27:14. Schloss
  Dagstuhl--Leibniz-Zentrum fuer Informatik, 2019.
\newblock \href {https://doi.org/10.4230/LIPIcs.ICALP.2019.27}
  {\path{doi:10.4230/LIPIcs.ICALP.2019.27}}.

\bibitem[BS83]{BS83}
Walter Baur and Volker Strassen.
\newblock The complexity of partial derivatives.
\newblock {\em Theoretical Computer Science}, 22(3):317--330, 1983.
\newblock \href {https://doi.org/10.1016/0304-3975(83)90110-X}
  {\path{doi:10.1016/0304-3975(83)90110-X}}.

\bibitem[BS17]{BS17}
Fernando~G.S.L. Brand{\~a}o and Krysta~M. Svore.
\newblock Quantum speed-ups for solving semidefinite programs.
\newblock In {\em 58th Annual Symposium on Foundations of Computer Science
  ({FOCS} 2017)}, oct 2017.
\newblock \href {https://doi.org/10.1109/focs.2017.45}
  {\path{doi:10.1109/focs.2017.45}}.

\bibitem[BS18]{balkanski2018parallelization}
Eric Balkanski and Yaron Singer.
\newblock Parallelization does not accelerate convex optimization: Adaptivity
  lower bounds for non-smooth convex minimization.
\newblock {\em arXiv preprint arXiv:1808.03880}, 2018.
\newblock \href {http://arxiv.org/abs/1808.03880} {\path{arXiv:1808.03880}}.

\bibitem[Bub15]{Bub15}
S\'{e}bastien Bubeck.
\newblock Convex optimization: Algorithms and complexity.
\newblock {\em Found. Trends Mach. Learn.}, 8(3–4):231–357, November 2015.
\newblock \href {https://doi.org/10.1561/2200000050}
  {\path{doi:10.1561/2200000050}}.

\bibitem[CCLW20]{CCLW20}
Shouvanik Chakrabarti, Andrew~M. Childs, Tongyang Li, and Xiaodi Wu.
\newblock Quantum algorithms and lower bounds for convex optimization.
\newblock {\em {Quantum}}, 4:221, January 2020.
\newblock \href {https://doi.org/10.22331/q-2020-01-13-221}
  {\path{doi:10.22331/q-2020-01-13-221}}.

\bibitem[DG19]{diakonikolas2019lower}
Jelena Diakonikolas and Crist{\'o}bal Guzm{\'a}n.
\newblock Lower bounds for parallel and randomized convex optimization.
\newblock In {\em Proceedings of the Thirty-Second Conference on Learning
  Theory}, volume~99, pages 1132--1157. PMLR, 2019.
\newblock URL: \url{http://proceedings.mlr.press/v99/diakonikolas19c.html}.

\bibitem[GAW19]{GAW19}
Andr\'{a}s Gily\'{e}n, Srinivasan Arunachalam, and Nathan Wiebe.
\newblock Optimizing quantum optimization algorithms via faster quantum
  gradient computation.
\newblock In {\em Proceedings of the Thirtieth Annual ACM-SIAM Symposium on
  Discrete Algorithms}, SODA '19, page 1425–1444, 2019.
\newblock \href {https://doi.org/10.1137/1.9781611975482.87}
  {\path{doi:10.1137/1.9781611975482.87}}.

\bibitem[GW08]{GW08}
Andreas Griewank and Andrea Walther.
\newblock {\em Evaluating Derivatives: Principles and Techniques of Algorithmic
  Differentiation, Second Edition}.
\newblock Other Titles in Applied Mathematics. Society for Industrial and
  Applied Mathematics, 2008.
\newblock URL: \url{https://books.google.com/books?id=xoiiLaRxcbEC}.

\bibitem[Jor05]{Jor05}
Stephen~P. Jordan.
\newblock Fast quantum algorithm for numerical gradient estimation.
\newblock {\em Phys. Rev. Lett.}, 95:050501, Jul 2005.
\newblock \href {https://doi.org/10.1103/PhysRevLett.95.050501}
  {\path{doi:10.1103/PhysRevLett.95.050501}}.

\bibitem[KL18]{KL18}
Sham~M Kakade and Jason~D Lee.
\newblock Provably correct automatic sub-differentiation for qualified
  programs.
\newblock In {\em Advances in Neural Information Processing Systems 31 (NeurIPS
  2018)}, pages 7125--7135. 2018.
\newblock URL:
  \url{http://papers.nips.cc/paper/7943-provably-correct-automatic-sub-differentiation-for-qualified-programs}.

\bibitem[KP20]{kerenidis2020quantum}
Iordanis Kerenidis and Anupam Prakash.
\newblock Quantum gradient descent for linear systems and least squares.
\newblock {\em Physical Review A}, 101(2):022316, 2020.
\newblock \href {https://doi.org/10.1103/PhysRevA.101.022316}
  {\path{doi:10.1103/PhysRevA.101.022316}}.

\bibitem[Nem94]{Nem94}
A.~Nemirovski.
\newblock On parallel complexity of nonsmooth convex optimization.
\newblock {\em Journal of Complexity}, 10(4):451 -- 463, 1994.
\newblock \href {https://doi.org/10.1006/jcom.1994.1025}
  {\path{doi:10.1006/jcom.1994.1025}}.

\bibitem[Nes04]{Nes04}
Yurii Nesterov.
\newblock {\em Introductory Lectures on Convex Optimization}.
\newblock Springer {US}, 2004.
\newblock \href {https://doi.org/10.1007/978-1-4419-8853-9}
  {\path{doi:10.1007/978-1-4419-8853-9}}.

\bibitem[Nes18]{nesterov2018lectures}
Yurii Nesterov.
\newblock {\em Lectures on convex optimization}, volume 137.
\newblock Springer, 2018.
\newblock \href {https://doi.org/10.1007/978-3-319-91578-4}
  {\path{doi:10.1007/978-3-319-91578-4}}.

\bibitem[NY83]{nemirovsky1983problem}
Arkadi{\u\i} Nemirovsky and David~Borisovich Yudin.
\newblock {\em Problem complexity and method efficiency in optimization}.
\newblock Wiley, 1983.

\bibitem[RSW{\etalchar{+}}19]{RSW+19}
Patrick Rebentrost, Maria Schuld, Leonard Wossnig, Francesco Petruccione, and
  Seth Lloyd.
\newblock Quantum gradient descent and {N}ewton's method for constrained
  polynomial optimization.
\newblock {\em New Journal of Physics}, 21(7):073023, 2019.
\newblock \href {https://doi.org/10.1088/1367-2630/ab2a9e}
  {\path{doi:10.1088/1367-2630/ab2a9e}}.

\bibitem[WS17]{woodworth2017lower}
Blake Woodworth and Nathan Srebro.
\newblock Lower bound for randomized first order convex optimization.
\newblock {\em arXiv preprint arXiv:1709.03594}, 2017.
\newblock \href {http://arxiv.org/abs/1709.03594} {\path{arXiv:1709.03594}}.

\end{thebibliography}

\end{document}